\documentclass[a4paper,USenglish]{lipics-v2018-new}
\usepackage{microtype}%if unwanted, comment out or use option "draft"

\nolinenumbers

\title{Uniformization Problems for Synchronizations of Automatic Relations on Words}

\author{Sarah Winter}{RWTH Aachen University, Germany}{winter@cs.rwth-aachen.de}{}{}%mandatory, please use full name; only 1 author per \author macro; first two parameters are mandatory, other parameters can be empty.

\authorrunning{S. Winter}%mandatory. First: Use abbreviated first/middle names. Second (only in severe cases): Use first author plus 'et al.'

\Copyright{Sarah Winter}%mandatory, please use full first names. LIPIcs license is "CC-BY";  http://creativecommons.org/licenses/by/3.0/

% mandatory: Please choose ACM 2012 classifications from https://www.acm.org/publications/class-2012 or https://dl.acm.org/ccs/ccs_flat.cfm . E.g., cite as "General and reference $\rightarrow$ General literature" or \ccsdesc[100]{General and reference~General literature}. 

%mandatory

\category{}%optional, e.g. invited paper

\relatedversion{}%optional, e.g. full version hosted on arXiv, HAL, or other respository/website

\supplement{}%optional, e.g. related research data, source code, ... hosted on a repository like zenodo, figshare, GitHub, ...

\funding{Supported by the project (LO 1174/3-1) of the German Research Foundation (DFG).}%optional, to capture a funding statement, which applies to all authors. Please enter author specific funding statements as fifth argument of the \author macro.

\EventEditors{Ioannis Chatzigiannakis, Christos Kaklamanis, D\'{a}niel Marx, and Don Sannella}
\EventNoEds{4}
\EventLongTitle{45th International Colloquium on Automata, Languages, and Programming (ICALP 2018)}
\EventShortTitle{ICALP 2018}
\EventAcronym{ICALP}
\EventYear{2018}
\EventDate{July 9--13, 2018}
\EventLocation{Prague, Czech Republic}
\EventLogo{eatcs}
\SeriesVolume{107}
\ArticleNo{281}
\usepackage{hyperref}
\usepackage{bbm,amsmath}
\usepackage{stmaryrd}
\usepackage{thm-restate}
\usepackage{enumitem}
\usepackage{diagbox}
\usepackage[usenames,dvipsnames,svgnames,table]{xcolor}
\definecolor{dark-gray}{gray}{0.4}
\usepackage{tikz}
\usetikzlibrary{arrows,automata,backgrounds,decorations.pathmorphing,positioning,calc}
\tikzset{shorten >=1pt, >=stealth, auto, node distance=40, initial text=}
\theoremstyle{definition} %plain, remark, definition
\newtheorem{assumption}[theorem]{Assumption}

\makeatletter
\newcommand{\charfusion}[2]{%
  \def\ch@rfusion##1##2{%
    \ooalign{\hfil$##1#1$\hfil\cr\hfil$##2$\hfil\crcr}}%
      \mathop{%
      \vphantom{#1}%
      \mathpalette\ch@rfusion#2}\displaylimits}
\makeatother

\newcommand{\inp}{\mathbbmtt{i}}
\newcommand{\outp}{\mathbbmtt{o}}
\newcommand{\Qin}{Q^{\inp}}
\newcommand{\Qout}{Q^{\outp}}

\newcommand{\cupcdot}{\charfusion{\cup}{\cdot}}
\newcommand{\leqlag}[1]{\leq\!\!#1}

\newcommand{\llag}[1]{<\!\!#1}
\newcommand{\glag}[1]{>\!\!#1}

\newcommand{\val}[1]{\ensuremath{\mathit{val}_{#1}}}

\newcommand{\ro}[1]{\ensuremath{\mathit{root}({#1})}}

\begin{document}

\maketitle

\begin{abstract}
 A uniformization of a binary relation is a function that is contained in the relation and has the same domain as the relation.
 The synthesis problem asks for effective uniformization for classes of relations and functions that can be implemented in a specific way.
 
 We consider the synthesis problem for automatic relations over finite words (also called regular or synchronized rational relations) by functions implemented by specific classes of sequential transducers.
 
 It is known that the problem ``Given an automatic relation, does it have a uniformization by a subsequential transducer?'' is decidable in the two variants where the uniformization can either be implemented by an arbitrary subsequential transducer or it has to be implemented by a synchronous transducer.
 We introduce a new variant of this problem in which the allowed input/output behavior of the subsequential transducer is specified by a set of synchronizations and prove decidability for a specific class of synchronizations.
\end{abstract}

%%%%%%%%%%%%%%%%%%%%%%%%%%%%%%%%%%%%%%%%%%%%%%%%%%%%%%%%%%%%%%%%%%%%%%%%%%%
%%%%%%%%%%%%%%%%%%%%%%%%%%%%%%%%%%%%%%%%%%%%%%%%%%%%%%%%%%%%%%%%%%%%%%%%%%%
%%%%%%%%%%%%%%%%%%%%%%%%%%%%%%%%%%%%%%%%%%%%%%%%%%%%%%%%%%%%%%%%%%%%%%%%%%%
\section{Introduction}\label{sec:intro}

A uniformization of a binary relation is a function that selects for each element in the domain of the relation a unique image that is in relation with this element.
Of interest to us in this paper are uniformization problems in the setting where the relations and functions on words are defined by finite automata.
Relations on words defined by finite automata extend languages defined by finite automata.
Unlike for words, different finite automaton models for relations lead to different classes of relations.

Relations defined by asynchronous finite automata are referred to as rational relations.
An asynchronous finite automaton is a nondeterministic finite automaton with two tapes whose reading heads can move at different speeds.
An equivalent computation model are asynchronous finite transducers (see, e.g., \cite{berstel2009}), that is, nondeterministic finite automata whose transitions are labeled by pairs of words.

A well known subclass of rational relations are synchronized rational relations (see \cite{DBLP:journals/tcs/FrougnyS93}), which are defined by synchronous finite automata, that is, finite automata with two tapes whose reading heads move at the same speed.
Equivalently, we speak of definability by synchronous finite transducers.
The class of synchronized rational relations is also called automatic or regular, here, we use the term automatic.

One uniformization problem asks for proving that each relation in a given class has a certain kind of uniformization.
For example, each rational relation can be uniformized by an unambiguous rational function (see \cite{sakarovich:2009a}).
Here, we are interested in the decision version of the problem: Given a relation from some class, does it have a uniformization in some other class?
For the class of uniformizations we consider sequential transducers.
A sequential transducer reads the input word in a deterministic manner and produces a unique output word for each input word.

The sequential uniformization problem relates to the synthesis problem, which asks, given a specification that relates possible inputs to allowed outputs, whether there is a program implementing the specification, and if so, construct one.
This setting originates from Church's synthesis problem \cite{church1962logic}, where logical specifications over infinite words are considered. 
B{\"u}chi and Landweber \cite{buechi} showed that for specifications in monadic second order logic, that is, specifications that can be translated into synchronous finite automata, it is decidable whether it can be realized by a synchronous sequential transducer (see, e.g., \cite{Thomas08} for a modern presentation of this result).
Later, decidability has been extended to asynchronous sequential transducers \cite{hosch1972finite,holtmann2010degrees}. 

Going from the setting of infinite words to finite words uniformization by subsequential \footnote{A subsequential transducer can make a final output depending on the last state reached in a run whereas a sequential transducer can only produce output on its transitions.} transducers is considered.
The problem whether a relation given by a synchronous finite automaton can be realized by a synchronous subsequential transducer is decidable; this result can be obtained by adapting the proof from the infinite setting.
Decidability has been extended to subsequential transducers \cite{CarayolL14}.
Furthermore, for classes of asynchronous finite automata decidability results for synthesis of subsequential transducers have been obtained in \cite{FJLW16}.

A semi-algorithm in this spirit was introduced by \cite{johnson2010}, the algorithm is tasked to synthesize a subsequential transducer that selects the length lexicographical minimal output word for each input word from a given rational relation.

The decision problems that have been studied so far either ask for uniformization by a synchronous subsequential or by an arbitrary subsequential transducer.
Our aim is to study the decision problem: Given a rational relation, does it have a uniformization by a subsequential transducer in which the allowed input/output behavior is specified by a given language of synchronizations?
The idea is to represent a pair of words by a single word where each position is annotated over $\{1,2\}$ indicating whether it came from the input or output component.
The annotated string provides a synchronization of the pair.
It is known that the class of rational relations is synchronized by regular languages \cite{Niv-transductions-lc}.
More recently, main subclasses of rational relations have been characterized by their synchronizations \cite{conf/stacs/FigueiraL14}.

We show decidability for a given automatic relation and a given set of synchronizations that synchronizes an automatic relation.
Thus our decidability result generalizes the previously known decidability result for synthesis of synchronous subsequential transducers from automatic relations.

The paper is structured as follows.
First, in Sec.~\ref{sec:sync}, we fix our notations and recap characterizations of synchronization languages established in \cite{conf/stacs/FigueiraL14}.
In Sec.~\ref{sec:unifproblems}, we introduce uniformization problems with respect to synchronization languages and compare our setting with known results. 
In Sec.~\ref{sec:regular}, we prove decidability of the question whether an automatic relation has a uniformization by a subsequential transducer in which the input/output behavior is specified by a set of synchronizations that synchronizes an automatic relation.

Omitted proofs can be found in the appendix.

%%%%%%%%%%%%%%%%%%%%%%%%%%%%%%%%%%%%%%%%%%%%%%%%%%%%%%%%%%%%%%%%%%%%%%%%%%%
%%%%%%%%%%%%%%%%%%%%%%%%%%%%%%%%%%%%%%%%%%%%%%%%%%%%%%%%%%%%%%%%%%%%%%%%%%%
%%%%%%%%%%%%%%%%%%%%%%%%%%%%%%%%%%%%%%%%%%%%%%%%%%%%%%%%%%%%%%%%%%%%%%%%%%%
\section{Synchronizations of relations}\label{sec:sync}

Let $\mathbbm{N}$ denote the set of all non-negative integers $\{0,1,\dots\}$, and for every $k \in \mathbbm{N}$, let $\mathbf{k}$ denote the set $\{1,\dots,k\}$.
Given a finite set $A$, let $|A|$ denote its cardinality and $2^A$ its powerset.

\subparagraph*{Languages and relations of finite words.}
An \emph{alphabet} $\Sigma$ is a finite set of letters, a finite \emph{word} is a finite sequence over $\Sigma$.
The set of all finite words is denoted by $\Sigma^*$ and the empty word by $\varepsilon$.
The length of a word $w \in \Sigma^*$ is denoted by $|w|$, the number of occurrences of a letter $a \in \Sigma$ in $w$ by $\#_a(w)$. Given $w \in \Sigma^*$, $w[i]$ stands for the $i$th letter of $w$, and $w[i,j]$ for the subword $w[i]\dots w[j]$.

A \emph{language} $L$ over $\Sigma$ is a subset of $\Sigma^*$, and $\mathit{Pref}(L)$ is the set $\{ u \in \Sigma^* \mid \exists v: uv \in L \}$ of its prefixes.
The prefix relation is denoted by $\sqsubseteq$.
A \emph{relation} $R$ over $\Sigma$ is a subset of $\Sigma^* \times \Sigma^*$.
The \emph{domain} of a relation $R$ is the set $\textrm{dom}(R) = \{ u \mid (u,v) \in R\}$,
the \emph{image} of a relation $R$ is the set $\textrm{img}(R) = \{ v \mid (u,v) \in R\}$.
For $u \in \Sigma^*$, let $R(u) = \{ v \mid (u,v) \in R \}$ and write $R(u) = v$, if $R(u)$ is a singleton.

A \emph{regular expression $r$} over $\Sigma$ has the form $\emptyset$, $\varepsilon$, $\sigma \in \Sigma$, $r_1 \cdot r_2$, $r_1+r_2$, or $r_1^*$ for regular expressions $r_1$, $r_2$.
The term $r^+$ is short for $r\cdot r^*$.
The concatenation operator $\cdot$ is often omitted.
The language associated to $r$ is defined as usual, denoted $L(r)$, or conveniently, $r$.

\begin{definition}[synchronization, $L$-controlled \cite{conf/stacs/FigueiraL14}]\label{def:sync}
 For $c \in \{\inp,\outp\}$, referring to input and output, respectively, we define two morphisms $\pi_{c}\colon (\mathbf{2} \times \Sigma) \rightarrow \Sigma \cup \{\varepsilon\}$ by $\pi_{\inp}((i,a)) = a$ if $i=1$, otherwise $\pi_{\inp}((i,a)) = \varepsilon$, and likewise for $\pi_{\outp}$ with $i=2$.
 These morphisms are lifted to words over $(\mathbf{2} \times \Sigma)$.
 
 A word $w \in (\mathbf{2} \times \Sigma)^*$ is a \emph{synchronization} of a uniquely determined pair $(w_1,w_2)$ of words over $\Sigma$, where $w_1 = \pi_{\inp}(w)$ and  $w_2 = \pi_{\outp}(w)$. We write $\llbracket w \rrbracket$ to denote $(w_1,w_2)$.
 Naturally, a set $S \subseteq (\mathbf{2} \times \Sigma)^*$ of synchronizations defines the relation $\llbracket S \rrbracket = \{ \llbracket w \rrbracket \mid w \in S\}$.
 
 A word $w = (i_1,a_1)\dots (i_n,a_n) \in (\mathbf{2} \times \Sigma)^*$ is the convolution $u \otimes v$ of two words $u = i_1\dots i_n \in \mathbf{2}^*$ and $v = a_1\dots a_n \in \Sigma^*$.
 Given a language $L \subseteq \mathbf{2}^*$, we say $w$ is \emph{$L$-controlled} if $u \in L$.
 A language $S \subseteq (\mathbf{2} \times \Sigma)^*$ is \emph{$L$-controlled} if all its words are.
 
 A language $L \subseteq \mathbf{2}^*$ is called a \emph{synchronization language}.
 For a regular language $L \subseteq \mathbf{2}^*$, $\textsc{Rel}(L) = \{ \llbracket S \rrbracket \!\mid\! S \text{ is a regular $L$-controlled }$ $\text{language}\}$ is the set of relations that can be given by $L$-controlled synchronizations.
 Let $\mathcal C$ be a class of relations, we say $L$ \emph{synchronizes} $\mathcal C$ if $\textsc{Rel}(L) \subseteq \mathcal C$.
 \end{definition}

\begin{definition}[lag, shift, shiftlag \cite{conf/stacs/FigueiraL14}]
Given a word $w \in \mathbf{2}^*$, a position $i \leq |w|$, and $\gamma \in \mathbbm{N}$.
We say $i$ is \emph{$\gamma$-lagged} if $|\#_1(w[1,i]) - \#_2(w[1,i])| = \gamma$, and likewise, we define \emph{$\glag{\gamma}$-lagged} and \emph{$\llag{\gamma}$-lagged}.
A \emph{shift} of $w$ is a position $i \in \{1,\dots,|w|-1\}$ such that $w[i] \neq w[i+1]$.
Two shifts $i < j$ are \emph{consecutive} if there is no shift $l$ such that $i < l < j$.
Let $\mathit{shift}(w)$ be the number of shifts in $w$, let $\mathit{lag}(w)$ be the maximum lag of a position in $w$, and let $\mathit{shiftlag}(w)$ be the maximum $n \in \mathbbm{N}$ such that $w$ contains $n$ consecutive shifts which are $\glag{n}$-lagged.

We lift these notions to languages by taking the supremum in $\mathbbm N \cup \{\infty\}$, e.g., $\mathit{shift}(L) = \mathrm{sup}\{\mathit{shift}(w) \mid w \in L\}$, and likewise for $\mathit{lag}(L)$ and $\mathit{shiftlag}(L)$.
\end{definition}

The following characterizations for well known subclasses of rational relations were shown in \cite{conf/stacs/FigueiraL14}.
Recall, rational relations are definable by asynchronous finite automata, automatic relations by synchronous finite automata, and recognizable relations are definable as finite unions of products of regular languages.
We omit a formal definition of these models since it is not relevant to this paper.

\begin{theorem}[\cite{conf/stacs/FigueiraL14}]
 Let $L \subseteq \mathbf{2}^*$ be a regular language. Then:
 \begin{enumerate}[noitemsep,topsep=0pt]
  \item $L$ synchronizes recognizable relations iff $\mathit{shift}(L) < \infty$,
  \item $L$ synchronizes automatic relations iff $\mathit{shiftlag}(L) < \infty$,
%  \item $L$ synchronizes relations in $\mathsf{REG}^{\mathit{bld}}$ iff $\mathit{lag}(L) < \infty$,
  \item $L$ synchronizes rational relations.
 \end{enumerate}
\end{theorem}

For ease of presentation, let $\Sigma_{\inp\outp}$, $\Sigma_{\inp}$, and $\Sigma_{\outp}$ be short for $\mathbf{2} \times \Sigma$, $\{1\} \times \Sigma$, and $\{2\} \times \Sigma$, respectively.
If convenient, we use distinct symbols for input and output, instead of symbols annotated with $1$ or $2$.

For the results shown in this paper, it is useful to lift some notions introduced in \cite{conf/stacs/FigueiraL14} from words and languages over $\mathbf{2}$ to words and languages over $\Sigma_{\inp\outp}$.

\begin{definition}
We lift the notions of $\mathit{lag}$, $\mathit{shift}$, and $\mathit{shiftlag}$ from words and languages over $\mathbf{2}$ to words and languages over $\Sigma_{\inp\outp}$ in the natural way.

Furthermore, given a language $T \subseteq \Sigma_{\inp\outp}^*$, we say a word $w \in \Sigma_{\inp\outp}^*$ is $T$-controlled if $w \in T$.
A language $S \subseteq \Sigma_{\inp\outp}^*$ is $T$-controlled if all its words are, namely, if $S \subseteq T$.
\end{definition}

\subparagraph*{Automata on finite words.}
We fix our notations concerning finite automata on finite words.
A \emph{nondeterministic finite automaton} (\emph{NFA}) is a tuple $\mathcal A = (Q,\Sigma,q_0,\Delta,F)$, where $Q$ is a finite set of states, $\Sigma$ is a finite alphabet, $q_0 \in Q$ is the initial state, $\Delta \subseteq Q \times \Sigma \times Q$ is the transition relation, and $F \subseteq Q$ is the set of final states.
A \emph{run} $\rho$ of $\mathcal A$ on $w = a_1\dots a_n \in \Sigma^*$ is a sequence of states $p_0p_1\dots p_n$ such that $(p_i,a_{i+1},p_{i+1}) \in \Delta$ for all $i \in \{0,\dots,n-1\}$.
Shorthand, we write $\mathcal A: p_0 \xrightarrow{w} p_n$.
A run is \emph{accepting} if it starts in $q_0$ and ends in a state from $F$.
The language \emph{recognized} by $\mathcal A$, written $L(\mathcal A)$, is the set of words $w \in \Sigma^*$ that admit an accepting run of $\mathcal A$ on $w$.
For $q \in Q$, let $\mathcal A_q$ denote the NFA obtained from $\mathcal A$ by setting its initial state to $q$.
The class of languages recognized by NFAs is the class of regular languages.
An NFA is \emph{deterministic} (a \emph{DFA}) if for each state $q \in Q$ and $a \in \Sigma$ there is at most one outgoing transition.
In this case, it is more convenient to express $\Delta$ as a (partial) function $\delta: Q \times \Sigma \rightarrow Q$.
Furthermore, let $\delta^*$ denote the usual extension of $\delta$ from letters to words.

We introduce some notions only applicable if an NFA recognizes a set of synchronizations.
Given a regular $S \subseteq \Sigma_{\inp\outp}^*$, let $\mathcal A = (Q,\Sigma_{\inp\outp},q_0,\Delta,F)$ be an NFA that recognizes $S$.
We define $\Qin = \{ p \in Q \mid \exists a \in \Sigma, q \in Q: (p,(1,a),q) \in \Delta\}$ and $\Qout = \{ p \in Q \mid \exists a \in \Sigma, q \in Q: (p,(2,a),q) \in \Delta\}$ as the sets of states that have outgoing transitions from which input and output can be consumed, respectively.
If $(\Qin$,$\Qout)$ is a partition of $Q$, we write $Q = \Qin \cupcdot \Qout$.
We call $\mathcal A$ \emph{sequential} if $\mathcal A$ is deterministic, and $Q = \Qin \cupcdot \Qout$, and each $q \in \Qout$ has at most one outgoing transition.
For short, we refer to a sequential DFA as \emph{sDFA}.
Finally, we define the input automaton $\mathcal A_D$ of $\mathcal A$ as $(Q,\Sigma,q_0,\Delta',F)$, where $\Delta' = \{ (p,a,q) \mid \mathcal A: p \xrightarrow{w} q \text{ and } \pi_{\inp}(w) = a \in \Sigma\}$.
A comparison to standard transducer models is given in the next section. 

%%%%%%%%%%%%%%%%%%%%%%%%%%%%%%%%%%%%%%%%%%%%%%%%%%%%%%%%%%%%%%%%%%%%%%%%%%%
%%%%%%%%%%%%%%%%%%%%%%%%%%%%%%%%%%%%%%%%%%%%%%%%%%%%%%%%%%%%%%%%%%%%%%%%%%%
%%%%%%%%%%%%%%%%%%%%%%%%%%%%%%%%%%%%%%%%%%%%%%%%%%%%%%%%%%%%%%%%%%%%%%%%%%%
\section{Uniformization problems}\label{sec:unifproblems}

A \emph{uniformization} of a relation $R \subseteq \Sigma^* \times \Sigma^*$ is a complete function $f_R: \mathrm{dom}(R) \to \Sigma^*$ with $(u,f_R(u)) \in R$ for all $u \in \mathrm{dom}(R)$.
If such a function is given as a relation $R_f$, we write $R_f \subseteq_{\mathsf{u}} R$ to indicate that $R_f$ is a uniformization of $R$.

\begin{definition}[Resynchronized uniformization problem]
 The \emph{resynchronized uniformization problem} asks, given a regular source language $S \subseteq \Sigma_{\inp\outp}^*$ and a regular target language $T \subseteq \Sigma_{\inp\outp}^*$, whether there exists a regular language $U \subseteq T$ recognized by a sequential DFA such that $\llbracket U \rrbracket \subseteq_{\mathsf{u}} \llbracket S \rrbracket$.
\end{definition}

%We start by giving an example.

\begin{example}\label{ex:intro}
 Let $\Sigma_\inp = \{a,b,c\}$ and $\Sigma_\outp = \{d,e\}$, let $S \subseteq \Sigma_{\inp\outp}^*$ be given by $\mathcal A$ depicted in Fig.~\ref{fig:intro}.
 The recognized relation is $\llbracket S \rrbracket = \{(a^iba^j,d(d+e)^k) \mid i,j,k \geq 0 \} \cup \{(a^ica^j,e(d+e)^k) \mid i,j,k \geq 0 \}$. 
 Furthermore, let $T = \Sigma_\inp^*(\Sigma_\inp\Sigma_\outp)^+$.
 A $T$-controlled uniformization $U$ is given by the sequential DFA $\mathcal U$ depicted in Fig.~\ref{fig:intro}.
 The recognized relation is $\llbracket U \rrbracket = \{(a^iba^j,dd^j) \mid i,j,k \geq 0 \} \cup \{(a^ica^j,ed^j) \mid i,j \geq 0 \}$.
\end{example}
\vskip -0.3cm

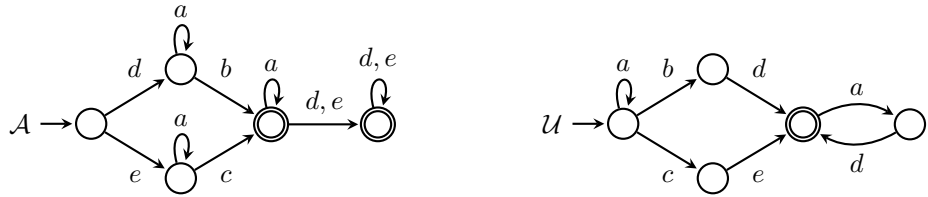
\begin{figure}[ht]
\vskip -0.35cm
\centering
  \begin{tikzpicture}[thick]
    \tikzstyle{every state}+=[inner sep=4pt, minimum size=3pt];
    \node[state, initial, initial text=$\mathcal A$] (0) {};
    \node[state, above right = 12pt and 25pt of 0] (1) {};
    \node[state, below right = 12pt and 25pt of 0] (2) {};
    \node[state, accepting, below right = 12pt and 25pt of 1] (3) {};
    \node[state, accepting, right of=3]       (4) {};
    \draw[->] (0) edge                node[near end]        {$d$} (1);
    \draw[->] (0) edge                node[swap,near end]  {$e$} (2);
    \draw[->] (1) edge[loop above]    node        {$a$} ();
    \draw[->] (1) edge                node[near start]        {$b$} (3);
    \draw[->] (2) edge[loop above]    node        {$a$} ();
    \draw[->] (2) edge                node[swap,near start]  {$c$} (3);
    \draw[->] (3) edge[loop above]    node        {$a$} ();
    \draw[->] (3) edge                node        {$d,e$} (4);
    \draw[->] (4) edge[loop above]    node        {$d,e$} ();

    \begin{scope}[xshift=7cm]
    \node[state, initial, initial text=$\mathcal U$] (0) {};
    \node[state, above right = 12pt and 25pt of 0] (1) {};
    \node[state, below right = 12pt and 25pt of 0] (2) {};
    \node[state, accepting, below right = 12pt and 25pt of 1] (3) {};
    \node[state,  right of=3]       (4) {};
    \draw[->] (0) edge[loop above]    node        {$a$} ();
    \draw[->] (0) edge                node[near end]        {$b$} (1);
    \draw[->] (0) edge                node[swap,near end]  {$c$} (2);
    \draw[->] (1) edge                node[near start]        {$d$} (3);
    \draw[->] (2) edge                node[swap,near start]  {$e$} (3);
    \draw[->] (3) edge[bend left]     node        {$a$} (4);
    \draw[->] (4) edge[bend left]     node        {$d$} (3);   
    \end{scope}
   \end{tikzpicture}
  \caption{
   Cf.~Ex.\ref{ex:intro};
   $S = L(\mathcal A)$ and $U = L(\mathcal U)$, we have $\llbracket U \rrbracket \subseteq_\mathsf{u} \llbracket S \rrbracket$.
   }
 \label{fig:intro}
\vskip -0.3cm
\end{figure} 

Comparing our definition of sequential DFAs with standard transducer models we notice that sequential transducers directly correspond to sequential DFAs.
See, e.g., \cite{berstel2009} for an introduction to transducers.
Our model can be modified to correspond to subsequential transducers (which can make a final output after the word has ended) by slightly modifying the representation of the relation by adding a dedicated endmarker in the usual way.

In the remainder it is implicitly assumed that every given source and target language is represented with endmarkers, thus our stated results correspond to uniformization by subsequential transducers.

Our main result is the decidability of the resynchronized uniformization problem for a given automatic relation and a given set of synchronizations controlled by a language that synchronizes automatic relations.
In Sec.~\ref{sec:regular} we see that our decidability result is obtained by a reduction to the following simpler uniformization problem.

\begin{definition}[Subset uniformization problem]
 The \emph{subset uniformization problem} asks, given a regular language $S \subseteq \Sigma_{\inp\outp}^*$, whether there exists a regular language $U \subseteq S$ recognized by a sequential DFA such that $\llbracket U \rrbracket \subseteq_{\mathsf{u}} \llbracket S \rrbracket$.
\end{definition}

The notion of subset uniformization directly corresponds to the notion of sequential $\mathbbm I$-uniformization introduced in \cite{FJLW16}.
It was shown that deciding the sequential $\mathbbm I$-uniformization problem reduces to deciding which player has a winning strategy in a safety game between $\mathsf{In}$ and $\mathsf{Out}$.
Hence, we directly obtain the following result.

\begin{restatable}[\cite{FJLW16}]{theorem}{thmsunif}\label{thm:sunif}
 The subset uniformiza\-tion problem is decidable.
\end{restatable}

Now that we have formulated our uniformization problems, we link these to known uniformization problems.
Asking whether a relation has a $\Sigma_{\inp\outp}^*$-controlled subsequential uniformization is equivalent to asking whether it has a uniformization by an arbitrary subsequential transducer.
Asking whether a relation has a $(\Sigma_\inp\Sigma_\outp)^*(\Sigma_\inp^* + \Sigma_\outp^*)$- resp.\ $\Sigma_\inp^*\Sigma_\outp^*$-controlled subsequential uniformization is equivalent to asking whether it has a uniformization by a synchronous subsequential transducer resp.\ by a transducer that reads the complete input before producing output.

\begin{table}[t]
\begin{center}
\begin{tabular}{|l|c|c|c|c|c|} \hline
 \backslashbox[33mm]{sync.}{relation} & rational & \parbox[c]{2cm}{deterministic\\rational} & finite-valued & automatic & \parbox[c]{1.1cm}{recog-\\nizable} \\  \hline
 $\Sigma_{\inp\outp}^*$ 					& undec. \cite{CarayolL14} & dec. \cite{FJLW16} & dec. \cite{FJLW16} & dec. \cite{CarayolL14} & dec.\\  \hline
 $(\Sigma_\inp\Sigma_\outp)^*(\Sigma_\inp^* + \Sigma_\outp^*)$ 	& undec. \cite{CarayolL14} & ? & ? & dec. \cite{buechi} & dec.\\  \hline
 $\Sigma_\inp^*\Sigma_\outp^*$ 					& ?   		   & ? & ? & dec. \cite{CarayolL14} & dec.\\  \hline \hline
 rational 		& undec. & ? & ? &  ?  & dec.\\  \hline
 automatic 		& undec. & ? & ? & \textbf{dec.}  & dec.\\  \hline
 recognizable 		&   ?    & ? & ? & dec.  & dec. \\  \hline
\end{tabular}
\end{center}
\caption{Overview over decidability results. 
The columns list the type of relation to be uniformized.
The rows list the type of synchronization used as uniformization parameter; the upper three rows list fixed languages of synchronizations, the lower three rows list parameter classes, where `rational' means the given set of allowed synchronizations is controlled by an arbitrary synchronization language, `automatic' (resp.\ `recognizable') means the given set of allowed synchronizations is controlled by a synchronization language that synchronizes automatic (resp.\ recognizable) relations.
}
\label{tab:overview}
\vskip -0.8cm
\end{table}

Table~\ref{tab:overview} provides an overview over known and new decidability results of the resynchronized uniformization problem for different types of relations and synchronization parameters.
Our main result is the decidability for a given automatic relation and a given set of allowed synchronizations that is controlled by a synchronization language that synchronizes automatic relations.
The decidability results in the rightmost column can be shown by a simple reduction to the subset uniformization problem which is presented in the appendix.
The other entries in the lower three rows are simple consequences of the results presented in the upper three rows resp.\ our main result.

Regarding the table entry where the relation is automatic and a desired uniformizer is $(\Sigma_\inp\Sigma_\outp)^*(\Sigma_\inp^* + \Sigma_\outp^*)$-controlled, there is an alternative formulation of the decision problem in the case that the given relation is $(\Sigma_\inp\Sigma_\outp)^*(\Sigma_\inp^* + \Sigma_\outp^*)$-controlled (the usual presentation for automatic relations, e.g., by a synchronous transducer).
In this case the problem can also be stated as the question whether the relation has a subset uniformization.

We now generalize this to Parikh-injective synchronization languages.
Given some $L \subseteq \mathbf{2}^*$, let $\Pi_L: L \to \mathbbm{N}^2$ be the function that maps a word $w \in L$ to its \emph{Parikh image}, that is to the vector $(\#_1(w),\#_2(w))$.
We say $L$ is \emph{Parikh-injective} if $\Pi_L$ is injective.

\begin{restatable}{proposition}{thmparikh}\label{thm:parikh}
 Let $L \subseteq \mathbf{2}^*$ be a regular Parikh-injective language, let $S \subseteq \Sigma_{\inp\outp}^*$ be an $L$-controlled regular language and let $T = \{ w \in \Sigma^* \mid w \text{ is $L$-controlled}\}$.
 Every $T$-controlled uniformization of $S$ is a subset uniformization of $S$.
\end{restatable}

Given $L$, $S$ and $T$ as in Proposition \ref{thm:parikh}, it directly follows that the resynchronized uniformization problem is equivalent to the subset uniformization problem, which is decidable by Theorem \ref{thm:sunif}.

\section{Automatic uniformizations of automatic relations}\label{sec:regular}

Here we present our main result stating that it is decidable whether a given automatic relation has a uniformization by a subsequential transducer whose induced set of synchronizations is controlled by a given regular language that synchronizes automatic relations.

\begin{restatable}{theorem}{thmregular}\label{thm:regular}
 Given a regular source language with finite $\mathit{shiftlag}$ and a regular target language with finite $\mathit{shiftlag}$.
 Then, the resynchronized uniformization problem is decidable.
\end{restatable}

In \cite{conf/stacs/FigueiraL14}, it is shown that $(12)^*(1^*+2^*)$ is an effective canonical representative of the class $\mathsf{RL}_{\mathit{FSL}}$ of regular languages with finite $\mathit{shiftlag}$.
Meaning that for every $L \in \mathsf{RL}_{\mathit{FSL}}$ and every $R \in \textsc{Rel}(L)$, there is an effectively constructible $(12)^*(1^*+2^*)$-controlled regular language $S$ so that $\llbracket S \rrbracket = R$.

In the remainder of this section, let $S \subseteq \Sigma_{\inp\outp}^*$ be a regular source language with finite $\mathit{shiftlag}$.
Also, let $S_\mathit{can}$ be the equivalent $(12)^*(1^*+2^*)$-controlled language with $\llbracket S_\mathit{can} \rrbracket = \llbracket S \rrbracket$. % recognized by a DFA $\mathcal A$.
Furthermore, let $T \subseteq \Sigma_{\inp\outp}^*$ be a regular target language with finite $\mathit{shiftlag}$.

\begin{assumption}\label{asm:shiftlag}
 We assume that $S_\mathit{can}$ is recognized by a DFA $\mathcal A = (Q_\mathcal A,\Sigma_{\inp\outp},q_0^\mathcal A,\Delta_\mathcal A,F_\mathcal A)$, $T$ is recognized by a DFA $\mathcal B = (Q_\mathcal B,\Sigma_{\inp\outp},q_0^\mathcal B,\Delta_\mathcal B,F_\mathcal B)$ and $\mathit{shiftlag}(T) < n$.
\end{assumption}

For notational convenience, given $x \in \Sigma_\inp^*$ and $y \in \Sigma_\outp^*$, we write $\delta_\mathcal A^*(q,(x,y))$ to mean $\delta_\mathcal A^*(q,w)$, where $w \in \Sigma_{\inp\outp}$ is the canonical synchronization of $x$ and $y$, i.e., $w$ is the $(12)^*(1^*+2^*)$-controlled synchronization of the pair $(x,y)$.
\\

The remainder of this section is devoted to the proof of Theorem~\ref{thm:regular}.
The proof is split in two main parts; the goal of the first part is to show that if $S$ has a $T$-controlled uniformization by an sDFA, then $S$ has a $T_k$-controlled uniformization by an sDFA for a regular $T_k \subseteq T$ that is less complex than $T$, cf.\ Lemma~\ref{lemma:shortregular}.
The goal of the second part is to show that the set $T_k(S)$ defined by $\{ w \mid w \in T_k \text{ and } \llbracket w \rrbracket \in \llbracket S \rrbracket\}$ is regular and computable (due to the form of $T_k$), cf.\ Lemma~\ref{lemma:transformregular}.
Then, to conclude the proof, we show that the question whether $S$ has a $T$-controlled uniformization by an sDFA can be reduced to the question whether $T_k(S)$ has a subset uniformization by an sDFA, which is decidable by Theorem~\ref{thm:sunif}.

Towards giving an exact description of $T_k$, consider the following auxiliary lemma characterizing the form of regular synchronization languages with finite $\mathit{shiftlag}$.
Given $\nu \in \mathbbm{N}$, we denote by $L_{\leq \nu}$ the regular set of words over $\mathbf{2}$ with $\leqlag{\nu}$-lagged positions, i.e., $L_{\leq \nu} = \{ u \in \mathbf 2^* \mid \mathit{lag}(u) \leq \nu\}$; we denote by $T_{\leq \nu}$ the regular set of words over $\Sigma_{\inp\outp}$ with $\leqlag{\nu}$-lagged positions, i.e., $T_{\leq \nu} = \{ w \in \Sigma_{\inp\outp}^* \mid \mathit{lag}(w) \leq \nu\}$.

\begin{lemma}[\cite{conf/stacs/FigueiraL14}]\label{lemma:form}
 Given a regular language $L \subseteq \mathbf{2}^*$ with $\mathit{shiftlag}(L) < m$.
 It holds that $L \subseteq L_{\leq \nu} \cdot (1^*+2^*)^m$ with $\nu$ chosen as $2\left(m(|Q|+1)+1\right)$, where $Q$ is the state set of an NFA recognizing $L$.
\end{lemma}

Clearly, this lemma can be lifted to regular languages over $\Sigma_{\inp\outp}$.
Based on Asm.~\ref{asm:shiftlag} and Lemma~\ref{lemma:form}, we can make the following assumption.

\begin{assumption}\label{asm:nshift}
 Assume that $T \subseteq T_{\leq \gamma} \cdot (\Sigma_\inp^*+\Sigma_\outp^*)^n$ with $\gamma = 2\left(n(|Q_\mathcal B|+1)+1\right)$. % computed as in Lemma~\ref{lemma:form}.
\end{assumption}

Now, we can be more specific about $T_k \subseteq T$.

\begin{definition}\label{def:M}
 For $i \geq 0$, let $T_i$ be the set $T \cap \left (T_{\leq \gamma} \cdot (\Sigma_\inp^*+\Sigma_\outp^{\leq i})^n\right )$, that is, the set of $w \in 
 T$ such that after a position in $w$ is more than $\gamma$-lagged, the number of output symbols per block is at most $i$.
\end{definition}

Our aim is to show that there is a bound $k$ such that $S$ has either a $T_k$-controlled uniformization by an sDFA or no $T$-controlled uniformization by an sDFA.
From now on, we call an sDFA implementing a uniformization simply a uniformizer.

The main difficulty in solving the resynchronized uniformization problem is that in general a uniformizer can have unbounded lag, because the waiting time between shifts can be arbitrarily long.
The key insight for the proof is that if such a long waiting time for a shift from input to output is necessary, then, in order to determine the next output block, it is not necessary to store the complete input that is ahead.
We show that it suffices to consider an abstraction of the input that is ahead.
Therefore we will introduce input profiles based on state transformation trees we define below.

Similarly, to deal with the situation where there is a long waiting time for a shift from output to input, we introduce output profiles as an abstraction of output that is ahead.

The bound on the length of output blocks will be chosen based on the profiles.
Before defining profiles, we introduce some necessary definitions and notions.

\subparagraph*{Trees.}

A \emph{finite unordered unranked tree} over an alphabet, a tree for short, is a finite non-empty directed graph with a distinguished root node, such that for any node, there exists exactly one path from the root to this node.
Additionally, a mapping from the nodes of the graph to the alphabet is given.
More formally, a \emph{tree} $t$ over $\Sigma$ is given by a tuple $(V_t,E_t,v_t,\val{t})$, where $V_t$ is a non-empty set of nodes, $E_t \subseteq V_t \times V_t$ is a set of edges, $v_t$ is the root of $t$, also denoted $\ro{t}$, and $\val{t}$ is a mapping $V_t \to \Sigma$.
Furthermore, it is satisfied that any node is reached by a unique path from the root.
Let $T_\Sigma$ denote the set of all trees over $\Sigma$.
We only distinguish trees up to isomorphism.

Given a tree $t$ and a node $v$ of $t$, let $t|_v$ denote the \emph{subtree} of $t$ rooted at $v$.

An $a \in \Sigma$ can also be seen as a tree $a \in T_\Sigma$ defined by $(\{v\},\emptyset,v,\val{a})$, where $\val{a}(v) = a$.

For two trees $t_1$ and $t_2$ with $\val{t_1}(\ro{t_1}) = \val{t_2}(\ro{t_2})$, i.e., with the same root label, we define $t_1 \circ t_2$ as the tree $t$ given by $(V_t,E_t,\ro{t_1},\val{t})$, where $V_t = V_{t_1} \cup V_{t_2} \setminus \{\ro{t_2}\}$, $E_t = E_{t_1} \cup \{ (\ro{t},v) \mid (\ro{t_2},v) \in E_{t_2}\} \cup (E_{t_2} \setminus \{ (\ro{t_2},v) \in E_{t_2}\})$ and $\val{t}$ as $\val{t_1} \cup \val{t_2}$ over nodes in $V_t$ (assuming $V_{t_1} \cap V_{t_2} = \emptyset$).

Given $a \in \Sigma$ and trees $t_1,\dots,t_n$, we define $a(t_1\dots t_n)$ to be the tree $(V_t,E_t,\allowbreak r,\allowbreak \val{t})$, where $V_t = \bigcup_{i=1}^n V_{t_i} \cup \{r\}$ with a new node $r$, $E_t = \bigcup_{i=1}^n E_{t_i} \cup \{(r,\ro{t_i})\! \mid \allowbreak 1 \leq i \leq n\}$ and $\val{t}$ is defined as $\val{t}(r) = a$ and $\bigcup_{i=1}^n \val{t_i}$ (assuming $V_{t_i} \cap V_{t_j} = \emptyset$ for all $i \neq j$).

\subparagraph*{State transformation trees.}

Now that we have fixed our notations, we explain what kind of information we want to represent using state transformation trees.
Basically, for an input segment that is ahead and causes lag, we are interested in how the input segment can be combined with output segments of same or smaller length and how this output can be obtained.

In the following we give an intuitive example.

\begin{example}\label{ex:intuition}
 Let $\Sigma_\inp = \{a\}$ and $\Sigma_\outp = \{b,c\}$.
 Consider the language $S_1 \subseteq \Sigma_{\inp\outp}^*$ given by the DFA $\mathcal A_1$ depicted in Fig.~\ref{subfig:dfas}, and the language $T_1 \subseteq \Sigma_{\inp\outp}^*$ given by the DFA $\mathcal B_1$ depicted in Fig.~\ref{subfig:dfas}.
 %Our goal is to find a $T_1$-controlled uniformizer of $S_1$.
 As we can see, $S_1$ is $(12)^*(1^* + 2^*)$-controlled, thus, already in its canonical form, and $T_1$ is $1^*2^*1^*2^*$-controlled.
 Both languages have finite $\mathit{shiftlag}$.
 
 Generally, a $T_1$-controlled uniformizer of $S_1$ can have arbitrary large lag.
 We take a look at the runs starting from $q_0$ in $\mathcal A_1$ and starting from $p_0$ in $\mathcal B_1$ that the computation of such a uniformizer can induce.
 However, $\mathcal A_1$ can only be simulated on the part where the lag is recovered, but arbitrarily large lag can occur, thus our goal is to find an abstraction of the part that causes lag.
 E.g., assume that such a uniformizer reads $aa$ without producing output.
 Towards defining an abstraction of $aa$, we are interested in how $aa$ could be combined with outputs of same or smaller length and how these outputs could be produced by some $T_1$-controlled uniformizer.
 Such a uniformizer could read some more $a$s and eventually must produce output.
 Reading $a$s leads from $p_0$ to $p_1$ in $\mathcal B_1$. 
 There are a few possibilities how output of length at most two can be produced such that it is valid from $p_1$ and the simulation from $q_0$ can be continued.
 It is possible to output $b$ ($\delta^*_{\mathcal B_1}(p_1,b) = p_2$, $\delta^*_{\mathcal A_1}(q_0,aba)=q_1$), $bb$ ($\delta^*_{\mathcal B_1}(p_1,bb) = p_2$, $\delta^*_{\mathcal A_1}(q_0,abab)=q_0$) or $bc$ ($\delta^*_{\mathcal B_1}(p_1,bc) = p_2$, $\delta^*_{\mathcal A_1}(q_0,abac)=q_2$).
 Alternatively, it is possible to output $b$ ($\delta^*_{\mathcal B_1}(p_1,b) = p_2$, $\delta^*_{\mathcal A_1}(q_0,ab)=q_0$) read another $a$ ($\delta^*_{\mathcal B_1}(p_2,a) = p_3$) and then produce $b$ ($\delta^*_{\mathcal B_1}(p_3,b) = p_3$, $\delta^*_{\mathcal A_1}(q_0,ab)=q_0$) or $c$ ($\delta^*_{\mathcal B_1}(p_3,c) = p_3$, $\delta^*_{\mathcal A_1}(q_0,ac)=q_2$).
 We see that the outputs $bb$ and $bc$ can each be obtained in two different ways. % w.r.t.\ $\mathcal B_1$.
 Namely, as one single output block, or as two output blocks with an input block in between (w.r.t.\ $\mathcal B_1$, we do not care about the number of blocks w.r.t.\ $\mathcal A_1$).
 The maximal number of considered output blocks (w.r.t.\ the target synchronization) is parameterized in the formal definition.
  
 We take a look at the tree in Fig.~\ref{subfig:tree}, this tree contains all the state transformations that can be induced by the described possibilities.
 The possibilities to produce output in one single block is reflected by the edges $(v_0,v_1)$, $(v_0,v_2)$ and $(v_0,v_3)$ representing the state transformation induced by the respective output block.
 The possibilities to produce output in two blocks is reflected by the edges $(v_0,v_4)$ representing the state transformation induced by the first output block, $(v_4,v_5)$ representing the state transformation induced by the intermediate input block, $(v_5,v_6)$ and $(v_5,v_7)$ representing the state transformation induced by the respective second output block.
 
\end{example}

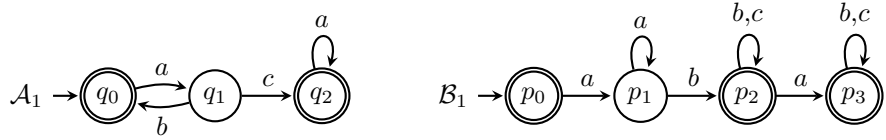
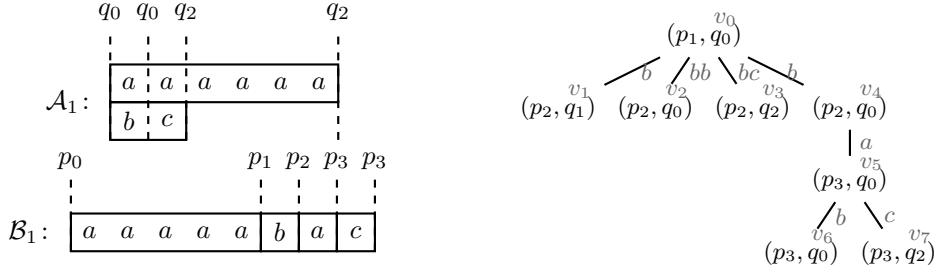
\begin{figure}[t!]
\vskip -0.5cm
\centering
\begingroup

\begin{subfigure}{\textwidth}
\begin{center}
\begin{tikzpicture}[scale=0.8,thick]
\tikzstyle{every state}+=[inner sep=3pt, minimum size=3pt];
    \node[state, accepting, initial, initial text=$\mathcal A_1$] (0) {$q_0$};
    \node[state, right of= 0] (1) {$q_1$};
    \node[state, accepting, right of= 1] (2) {$q_2$};

    \draw[->] (0) edge[bend left=15]  node        {$a$} (1);
    \draw[->] (1) edge[bend left=15]  node        {$b$} (0);
    \draw[->] (1) edge                node        {$c$} (2);
    \draw[->] (2) edge[loop above]    node        {$a$} ();

    \begin{scope}[xshift=7cm]
    \node[state, accepting, initial, initial text=$\mathcal B_1$] (0) {$p_0$};
    \node[state, right of=0] (1) {$p_1$};
    \node[state, accepting, right of=1] (2) {$p_2$};
    \node[state, accepting, right of=2] (3) {$p_3$};

    \draw[->] (0) edge                node        {$a$} (1);
    \draw[->] (1) edge                node        {$b$} (2);
    \draw[->] (1) edge[loop above]    node        {$a$} ();
    \draw[->] (2) edge                node        {$a$} (3);
    \draw[->] (2) edge[loop above]    node        {$b$,$c$} ();
    \draw[->] (3) edge[loop above]    node        {$b$,$c$} ();
    \end{scope}
\end{tikzpicture}
\end{center}
\vskip -1.3em
\caption{
$\Sigma_\inp = \{a\}$, $\Sigma_\outp = \{b,c\}$.
$\mathcal A_1$ recognizes $S_1$, $\mathcal B_1$ recognizes $T_1$.
$S_1$ is $(12)^*(1^* + 2^*)$-controlled and $T_1$ is $1^*2^*1^*2^*$-controlled, thus both have finite $\mathit{shiftlag}$. $S_1$ is already in the canonical form.
}
\label{subfig:dfas}
\end{subfigure}

\begin{subfigure}{0.49\textwidth}
\begin{center}
\begin{tikzpicture}[thick]

\tikzstyle{textshift}=[xshift=0.7em,yshift=-0.7em]

% \draw  (-5,5.5) rectangle (-2.5,5);
% \draw  (-2.5,5.5) rectangle (-2,5);
% \draw  (-5,5) rectangle (-4.5,4.5);
% \draw  (-4.5,5) rectangle (-4,4.5);

\draw  (-5,5.5) rectangle (-2,5);
\draw  (-5,5) rectangle (-4,4.5);

\node at (-5,5.5)   (a) {};
\node at (-4.5,5.5) (b) {};
\node at (-4,5.5)   (c) {};
\node at (-3.5,5.5) (d) {};
\node at (-3,5.5)   (e) {};
\node at (-2.5,5.5) (f) {};
\node at (-5,5)     (g) {};
\node at (-4.5,5)   (h) {};

\node[textshift] at (a) {$a$};
\node[textshift] at (b) {$a$};
\node[textshift] at (c) {$a$};
\node[textshift] at (d) {$a$};
\node[textshift] at (e) {$a$};
\node[textshift] at (f) {$a$};
\node[textshift] at (g) {$b$};
\node[textshift] at (h) {$c$};

\draw[dashed] ($ (g) - (0,0.5) $) -- ($ (g) + (0,1) $);
\draw[dashed] ($ (h) - (0,0.5) $) -- ($ (h) + (0,1) $);
\draw[dashed] ($ (-4,5) - (0,0.5) $) -- ($ (-4,5) + (0,1) $);
\draw[dashed] ($ (-2,5) - (0,0.5) $) -- ($ (-2,5) + (0,1) $);

\node at ($ (g) + (0,1.2) $) {$q_0$};
\node at ($ (h) + (0,1.2) $) {$q_0$};
\node at ($ (-4,5) + (0,1.2) $) {$q_2$};
\node at ($ (-2,5) + (0,1.2) $) {$q_2$};

\node at ($ (a) - (0.5,0.5) $) {$\mathcal A_1\colon$};

\end{tikzpicture}
\begin{tikzpicture}[thick]

\tikzstyle{textshift}=[xshift=0.7em,yshift=-0.7em]

\draw  (-5,5.5) rectangle (-2.5,5);
\draw  (-2.5,5.5) rectangle (-2,5);
\draw  (-2,5.5) rectangle (-1.5,5);
\draw  (-1.5,5.5) rectangle (-1,5);

\node at (-5,5.5)   (a) {};
\node at (-4.5,5.5) (b) {};
\node at (-4,5.5)   (c) {};
\node at (-3.5,5.5) (d) {};
\node at (-3,5.5)   (e) {};
\node at (-2.5,5.5) (f) {};
\node at (-2,5.5)   (g) {};
\node at (-1.5,5.5) (h) {};
\node at (-1,5.5)   (i) {};

\node[textshift] at (a) {$a$};
\node[textshift] at (b) {$a$};
\node[textshift] at (c) {$a$};
\node[textshift] at (d) {$a$};
\node[textshift] at (e) {$a$};
\node[textshift] at (f) {$b$};
\node[textshift] at (g) {$a$};
\node[textshift] at (h) {$c$};

\draw[dashed] ($ (a) - (0,0.5) $) -- ($ (a) + (0,0.5) $);
\draw[dashed] ($ (f) - (0,0.5) $) -- ($ (f) + (0,0.5) $);
\draw[dashed] ($ (g) - (0,0.5) $) -- ($ (g) + (0,0.5) $);
\draw[dashed] ($ (h) - (0,0.5) $) -- ($ (h) + (0,0.5) $);
\draw[dashed] ($ (i) - (0,0.5) $) -- ($ (i) + (0,0.5) $);

\node at ($ (a) + (0,0.7) $) {$p_0$};
\node at ($ (f) + (0,0.7) $) {$p_1$};
\node at ($ (g) + (0,0.7) $) {$p_2$};
\node at ($ (h) + (0,0.7) $) {$p_3$};
\node at ($ (i) + (0,0.7) $) {$p_3$};

\node at ($ (a) - (0.5,0.25) $) {$\mathcal B_1\colon$};

\end{tikzpicture}
\end{center}
\caption{Runs of $\mathcal A_1$ and $\mathcal B_1$ on synchronizations of $(aaaaaa,bc)$.
$\mathcal A_1$ runs on the canonical synchronization, i.e., on $abacaaaa$. To illustrate this, input and output are drawn one above the other.}
\label{subfig:runs}
\end{subfigure}
\begin{subfigure}{0.49\textwidth}
\begin{center}
\begin{tikzpicture}[thick,scale=0.8,baseline=(current bounding box.base)]
      \tikzstyle{level 1}=[sibling distance=16mm]
      \path[level distance=12mm] node (root){\small$(p_1,q_0)$}
	child{
	  node(0){\small$(p_2,q_1)$}
	}
	child{
	  node(1){\small$(p_2,q_0)$}
	}
	child{
	  node(2){\small$(p_2,q_2)$}
	}
	child{
	  node(3){\small$(p_2,q_0)$}
	  child{
	    node(4){\small$(p_3,q_0)$}
	    child{
	      node(5){\small$(p_3,q_0)$}
	    }
	    child{
	      node(6){\small$(p_3,q_2)$}
	    }
	  }
	}
      ;
      
      \path (root) -- coordinate[midway] (r0) (0);
      \path (root) -- coordinate[midway] (r1) (1);
      \path (root) -- coordinate[midway] (r2) (2);
      \path (root) -- coordinate[midway] (r3) (3);
      \path (3) -- coordinate[midway] (r4) (4);
      \path (4) -- coordinate[midway] (r5) (5);
      \path (4) -- coordinate[midway] (r6) (6);
      
      \node [dark-gray,right] at (r0) {\small$b$};
      \node [dark-gray,right] at (r1) {\small$bb$};
      \node [dark-gray,right] at (r2) {\small$bc$};
      \node [dark-gray,right] at (r3) {\small$b$};
      \node [dark-gray,right] at (r4) {\small$a$};
      \node [dark-gray,right] at (r5) {\small$b$};
      \node [dark-gray,right] at (r6) {\small$c$};
      
      \node [dark-gray,above right] at (root) {\small$v_0$};
      \node [dark-gray,above right] at (0) {\small$v_1$};
      \node [dark-gray,above right] at (1) {\small$v_2$};
      \node [dark-gray,above right] at (2) {\small$v_3$};
      \node [dark-gray,above right] at (3) {\small$v_4$};
      \node [dark-gray,above right] at (4) {\small$v_5$};
      \node [dark-gray,above right] at (5) {\small$v_6$};
      \node [dark-gray,above right] at (6) {\small$v_7$};
     \end{tikzpicture}
\end{center}
\caption{$\mathrm{STT}^1(aa,p_1,q_0)$.
The combination of both runs shown in Fig.~\ref{subfig:runs} is reflected by the rightmost path in the state transformation tree.}
\label{subfig:tree}
\end{subfigure}

 \caption{ 
 A source language $S_1$ and a target language $T_1$ are given in Fig.~\ref{subfig:dfas}. 
 A pair and two different synchronizations of said pair as well as runs are given in Fig.~\ref{subfig:runs}.
 The state transformation tree $\mathrm{STT}^1(aa,p_1,q_0)$ is given in Fig.~\ref{subfig:tree}, its edges are labeled with the respective associated words and its vertices are named for easier reference in Ex.~\ref{ex:intuition}. For a formal definition of STTs see Def.~\ref{def:inputstt}, for an explanation for this specific tree see Ex.~\ref{ex:intuition}.
 }
 \label{fig:inputstt}
 \endgroup
\vskip -0.5cm
\end{figure}

Now that we have given some intuition, we formally introduce input state transformation trees, a graphical representation of the construction of input state transformation trees is given in Fig.~\ref{fig:STT}.
As seen in the example, each edge of such a tree represents the state transformation induced by an output resp.\ input block, alternatively.

\begin{figure}
\vskip -0.5cm
 \centering
  \begin{tikzpicture}[thick,scale=0.95]

    \node[draw,circle,fill,inner sep=0pt,minimum size=3pt] (v2) at (-0.5,5) {};

    \node [draw,circle,fill,inner sep=0pt,minimum size=3pt] (v1) at (-3,4) {};

    \node [draw,circle,fill,inner sep=0pt,minimum size=3pt] (v3) at (-2.5,4) {};

    \node [draw,circle,fill,inner sep=0pt,minimum size=3pt] (v4) at (-2,4) {};

    \draw (v1) -- (v2);

    \draw (v3) -- (v2);

    \draw (v4) -- (v2);

    \draw [dark-gray] (-2.5,4) ellipse (1.5 and 0.25);

    \node [draw,circle,fill,inner sep=0pt,minimum size=3pt] (v10) at (-0.5,3.5) {};

    \node [draw,circle,fill,inner sep=0pt,minimum size=3pt] (v5) at (0,3.5) {};

    \node [draw,circle,fill,inner sep=0pt,minimum size=3pt] (v6) at (1,3.5) {};

    \node [draw,circle,fill,inner sep=0pt,minimum size=3pt] (v7) at (2.5,3.5) {};

    \draw (v2) -- (-0.5,3.5) node [draw,circle,fill,inner sep=0pt,minimum size=3pt] {};

    \draw (v2) -- (v5);

    \draw (v2) -- (v6);

    \draw (v2) -- (v7);

    \draw  [dark-gray](v6) ellipse (2 and 0.25);

    \node [draw,circle,fill,inner sep=0pt,minimum size=3pt] (v11) at (0,2) {};

    \node [draw,circle,fill,inner sep=0pt,minimum size=3pt] (v8) at (1,2) {};

    \node [draw,circle,fill,inner sep=0pt,minimum size=3pt] (v9) at (2.5,2) {};

    \draw (v8);

    \draw (1,2) -- (-1,0) -- (3,0) -- (v8);

    \draw (1,3.5) -- (0,2);

    \draw (v6) -- (v8);

    \draw (v6) -- (v9);

    \draw [dark-gray] (v8) ellipse (2 and 0.25);

    \draw [dotted](v10) -- (-5.5,0) -- (-4.5,0) -- (v10);

    \draw [dotted](v5) -- (-4,0) -- (-3,0) -- (v5);

    \draw [dotted](v7) -- (5,0) -- (6,0) -- (2.5,3.5);

    \draw [dotted](v11) -- (-2.5,0) -- (-1.5,0) -- (v11);

    \draw [dotted](v9) -- (3.5,0) -- (4.5,0) -- (v9);

    \node at (0,5) {\ \ \ \small$(p,q)$};

    \node at (1,0.5) {\small$\mathrm{STT}^{i-1}(x'',p'',q')$};

    \node at (1.5,3.5) {\ \ \ \small$(p',q')$};

    \node at (1.5,2) {\ \ \ \small$(p'',q')$};

    \node [dark-gray] at (-4,4.5) {\small$\mathrm{Reach}_0$};

    \node [dark-gray] at (3,4) {\small$\mathrm{Reach}_1$};

    \node [dark-gray] at (2.5,2.5) {\small$\mathrm{Reach}_{(x'',p',q')}$};

    %\node at (-4.5,5.5) {$\mathrm{STT}^{i}(x,p,q)$};
    
    \node [dark-gray, left] at (v2) {\small$v_0$};
    \node [dark-gray, left] at (v6){\small$v_1$};
    \node [dark-gray, left] at (v8) {\small$v_2$};

  \end{tikzpicture}
\caption{
 Schema of the input state transformation tree $\mathrm{STT}^{i}(x,p,q)$ for some $i > 0$.
 Cf.~Def.~\ref{def:inputstt}.
 Let $x'x''$ be a factorization of $x$ with $x', x'' \in \Sigma_\inp^+$, and let $y \in \Sigma_\outp^+$ be such that $|x'|= |y|$ and $\delta_\mathcal A^*(q,(x',y)) = q'$ and $\delta_\mathcal B^*(p,y) = p'$, and let $\delta_\mathcal B^*(p',w) = p''$ for some $w \in \Sigma_\inp^+$, then $\mathrm{STT}^{i}(x,p,q)$ contains a path $v_0v_1v_2$ labeled $(p,q)(p',q')(p'',q')$ such that $v_0$ is the root, $v_1$ is the root of $t^{i-1}_{(x'',p',q')}$, and $v_2$ is the root of $\mathrm{STT}^{i-1}(x'',p'',q')$.
}
\label{fig:STT}
\vskip -0.5cm
\end{figure}

\begin{definition}[Input state transformation tree]\label{def:inputstt}
For $i \geq 0$, $p \in Q_\mathcal B$, $q \in Q_\mathcal A$ and $x \in \Sigma_\inp^*$, the \emph{state transformation tree} $\mathrm{STT}^i(x,p,q)$ is a tree over $Q_\mathcal B \times Q_\mathcal A$ defined inductively.
 \begin{itemize}[topsep=1em]
 \item 
For $i = 0$, the tree $\mathrm{STT}^0(x,p,q)$ is built up as follows.

Let $\mathrm{Reach}_0 \subseteq Q_\mathcal B \times Q_\mathcal A$ be the smallest set such that
$(p',q') \in \mathrm{Reach}_0$ if there is some $y \in \Sigma_\outp^*$ with $|y| \leq |x|$ such that $\delta_\mathcal A^*(q,(x,y)) = q'$ and $\delta_\mathcal B^*(p,y) = p'$.

{\quad\small(This set represents state transformations induced by output blocks that fully consume $x$.)}

Then the tree $\mathrm{STT}^0(x,p,q)$ is defined as $(p,q)({r_1}\dots{r_n})$ 
% \begingroup
% \setlength{\abovedisplayskip}{.5\columnsep }
% \setlength{\belowdisplayskip}{.5\columnsep }
% \begin{equation*}
% (p,q)({r_1}\dots{r_n})
% \end{equation*}
% \endgroup
for $\mathrm{Reach}_0 = \{r_1,\dots,r_n\}$, meaning it contains a child for every state transformation that can be induced w.r.t.\ $\mathcal A$ and $\mathcal B$ starting from $q$ and $p$, respectively, by the input segment $x$ together with an output segment that consumes $x$ (w.r.t.\ $\mathcal A$) consisting of a single output block (w.r.t.\ $\mathcal B$). %, i.e., the output segment has zero cuts.

\item For $i > 0$, the tree $\mathrm{STT}^i(x,p,q)$ is built up as follows.

Let $\mathrm{Reach}_1 \subseteq \Sigma_\inp^* \times Q_\mathcal B \times Q_\mathcal A$ be the smallest set such that $(x'',p',q') \in \mathrm{Reach}_1$ if 
\begin{itemize}
 \item $x = x'x''$ with $x''\in \Sigma_\inp^+$ for an $x'\! \in \Sigma_\inp^+$ such that there is a $y \in \Sigma_\outp^+$ with $|y| = |x'|$, and
 \item $\delta_\mathcal A^*(q,(x',y)) = q'$ and $\delta_\mathcal B^*(p,y) = p'$.
\end{itemize}
\vskip -0.5em

{\quad\small(This set represents state transformations induced by output blocks that partially consume $x$.)}

For $(x'',p',q') \in \mathrm{Reach}_1$, let $\mathrm{Reach}_{(x'',p',q')} \subseteq \Sigma_\inp^* \times Q_\mathcal B \times Q_\mathcal A$ be the smallest set such that $(x'',p'',q') \in \mathrm{Reach}_{(x'',p',q')}$ if $\delta_\mathcal B^*(p',w) = p''$ for some $w \in \Sigma_\inp^+$.

{\quad\small(These sets represents state transformations induced by intermediate input blocks.)}

Furthermore, let the tree $t_{(x'',p',q')}^{i-1}$ be defined as $(p',q')(\mathrm{STT}^{i-1}{r_1}\dots\mathrm{STT}^{i-1}{r_n})$ for $\mathrm{Reach}_{(x'',p',q')}\allowbreak =\allowbreak \{r_1,\dots,r_n\}$.

Then the tree $\mathrm{STT}^i(x,p,q)$ is defined as % $\mathrm{STT}^0(x,p,q) \circ (p,q)(t_{s_1}^{i-1}\dots t_{s_n}^{i-1})$
\begingroup
\setlength{\abovedisplayskip}{.5\columnsep }
\setlength{\belowdisplayskip}{.5\columnsep }
\begin{equation*}
\mathrm{STT}^0(x,p,q) \circ (p,q)(t_{s_1}^{i-1}\dots t_{s_n}^{i-1})
\end{equation*}
\endgroup
for $\mathrm{Reach}_1 = \{s_1,\dots,s_n\}$, meaning it contains a path for every sequence of state transformations that can be induced w.r.t.\ $\mathcal A$ and $\mathcal B$ starting from $q$ and $p$, respectively, by the input segment $x$ together with an output segment that consumes $x$ (w.r.t.\ $\mathcal A$) consisting of at most $i+1$ output blocks (w.r.t.\ $\mathcal B$). %, i.e., the output segment has at most $i$ cuts.
Additionally, for output segments that have a common prefix of output blocks the state transformations induced by the common prefix of blocks are represented by the same nodes in the tree. %state transformation tree.
\end{itemize}
Intuitively, edges in such a tree are associated with the words that induced the state transformation, e.g., as shown in Fig~\ref{subfig:tree}.
\end{definition}

Given a tree as in Def.~\ref{def:inputstt}, the maximal degree of such a tree depends on the input word used as parameter.
Our goal is to have state transformation trees where the maximum degree is independent of this parameter.
Therefore, we introduce \emph{reduced trees}.
The idea is that if for some input word different outputs induce the same state transformations then only one representation is kept in the input state transformation tree. %, vice versa for output state transformation trees.

\begin{definition}[Reduced tree]\label{def:redtree}
 A tree $t \in T_{\Sigma}$ over some alphabet $\Sigma$ is called \emph{reduced} if for each node $v$ there exist no two children $u,u'$ of $v$ such that the subtrees rooted at $u$ and $u'$ are isomorphic.
 
 For a tree $t \in T_{\Sigma}$, let $\mathit{red}(t) \in T_\Sigma$ denote its reduced variant.
 The reduced variant of a tree can easily be obtained by a bottom-up computation where for each node duplicate subtrees rooted at its children are removed. 
\end{definition}

Note that for each $i$, the set of reduced input state transformation trees with parameter $i$ is a finite set.

Hitherto, we have discussed how to capture state transformations induced by an input word together with output words of same or smaller length.
Additionally, we need to capture state transformations induced by an output word together with input words of same or smaller length.
Therefore, we introduce a notion similar to input state transformation trees, namely, \emph{output state transformation trees}.
A formal definition can be found in the appendix.

%Thus far, we introduced the notion of state transformation trees to capture state transformations induced by the combination input and output words.
Furthermore, we need a notion that captures state transformations that can be induced by an input resp.\ output word alone, see Def.~\ref{def:stf} below.
Then, we are ready to define profiles.

\begin{definition}[State transformation function]\label{def:stf}
 For each $w \in \Sigma_\inp^* \cup \Sigma_\outp^*$, we define the function $\tau_w\colon Q_\mathcal B \to Q_\mathcal B$ with $\tau_w(p) = q$ if $\delta_\mathcal B^*(p,w) = q$ called \emph{state transformation function w.r.t.\ $w$}.
\end{definition}

%Now, we are ready to define profiles.

\subparagraph*{Profiles.}

Recall, $T \subseteq T_{\leq \gamma} \cdot (\Sigma_\inp^*+\Sigma_\outp^*)^n$, and our goal is to show that there is a bound $k$ such that it suffices to focus on constructing $T_k$-controlled uniformizers instead of $T$-controlled uniformizers, meaning that we can focus on uniformizers in which the length of output blocks is bounded by $k$ after the lag has exceeded $\gamma$ at some point.

The core of the proof is to show that if the lag between input and output becomes very large ($\gg \gamma$), it is not necessary to consider the complete input that is ahead to determine the next output block, but an abstraction (in the form of profiles) suffices.
Note that if the lag has exceeded $\gamma$ at some point the number of remaining output blocks is at most $\lceil n/2 \rceil$.

As a result, given an input word $x \in \Sigma_\inp^*$, we are interested in the state transformation that is induced by $(x,\pi_\outp(w))$ in $\mathcal A$ (recognizing $S_\mathit{can}$) and by $w$ in $\mathcal B$ (recognizing $T$) for each word $w \in \Sigma_{\inp\outp}^*$ such that $|\pi_\outp(w)| \leq |x|$ and $\mathit{shift}(w) \leq \lceil n/2 \rceil$.
In words, we are interested in the state transformations that can be induced by $x$ together with outputs of same or smaller length that are composed of at most $\lceil n/2 \rceil$ different output blocks.

For $x \in \Sigma_\inp^*$, this kind of information is accurately represented by the set of all reduced input state transformation trees with parameters $x$ and $\lceil n/2 \rceil$.

The same considerations with switched input and output roles apply for an output word $y \in \Sigma_\outp^*$.

%Now, we formally define profiles.

\begin{definition}[Input profile]
 For $x \in \Sigma_\inp^*$, we define its \emph{profile} $P_x$ as $(\tau_x,\mathrm{STT}_x^{\lceil n/2 \rceil})$, where 
 \begingroup
 \setlength{\abovedisplayskip}{.5\columnsep }
 \setlength{\belowdisplayskip}{.5\columnsep }
 \begin{equation*}
 \mathrm{STT}_x^{\lceil n/2 \rceil} = \bigcup_{(p,q) \in Q_\mathcal B \times Q_\mathcal A} \{\mathit{red}\bigl(\mathrm{STT}^{\lceil n/2 \rceil}(x,p,q)\bigr)\}.
 \end{equation*}
 \endgroup
\end{definition}

Similarly, we define \emph{output profiles}, a formal definition can be found in the appendix.

A note on the number of different profiles.
Profiles are based on reduced STTs with parameter $\lceil n/2 \rceil$, where $n$ bounds $\mathit{shiftlag(T)}$. 
The size of the set of these STTs is non-elementary in $n$, hence also the number of profiles.
This implies a non-elementary complexity of our decision procedure.

Furthermore, let $\mathcal P_\inp$ be the set $\bigcup_{x \in \Sigma_\inp^*} \{ P_x \}$ of all input profiles and $\mathcal P_\outp$ be the set $\bigcup_{y \in \Sigma_\outp^*} \{ P_y \}$ of all output profiles.
For a $P \in \mathcal P_\inp \cup \mathcal P_\outp$, let $z$ be a \emph{representative} of $P$ if $z$ is a shortest word such that $P = P_z$.

We show that from the profiles of two words $x_1$ and $x_2$ one can compute the profile of the word $x_1x_2$.
Hence, the set of profiles can be equipped with a concatenation operation, i.e., for words $x_1$ and $x_2$ we let $P_{x_1}P_{x_2} = P_{x_1x_2}$.
We obtain the following.

\begin{restatable}{lemma}{lemmamonoid}\label{lemma:monoid}
 The set of input profiles is a monoid with concatenation; the set of output profiles is a monoid with concatenation.
\end{restatable}

 A word $x \in \Sigma_\inp^*$ and its profile $P_x$ are called \emph{idempotent} if $P_x = P_{xx}$.
 As a consequence of Ramsey's Theorem (see e.g., \cite{diestel2000graduate}) we obtain the following lemma.

\begin{restatable}[Consequence of Ramsey]{lemma}{lemmaramsey}\label{lemma:ramsey}
 There is a computable $r \in \mathbbm{N}$ such that each word $x \in \Sigma_\inp^*$ with $|x| \geq r$ contains a non-empty idempotent factor for the concatenation of profiles.
\end{restatable}

Now, we have the right tools to prove that the existence of a $T$-controlled uniformizer implies that there also exists a $T_k$-controlled uniformizer for a computable $k$.
For the remainder, we fix two bounds.

\begin{assumption}\label{asm:bounds}
 Assume $r_1$ is chosen as in Lemma \ref{lemma:ramsey} and $r_2$ is chosen as the smallest bound on the length of representatives of output profiles.
 Wlog, assume $r_1,r_2 > \gamma$.
\end{assumption}

Finally, we are ready to prove the key lemma, that is, Lemma~\ref{lemma:shortregular}, which shows that it is sufficient to consider uniformizers in which the length of output blocks is bounded.

Recall, a uniformizer works asynchronously, which leads to large lag.
First, we show that if the output is lagged more than $r_1$ symbols, meaning, the input that is ahead contains an idempotent factor, it suffices to consider output blocks whose length depends on the idempotent factor.
Secondly, we show that it suffices to consider uniformizers in which the output is ahead at most $r_2$ symbols.
The combination of both results yields Lemma~\ref{lemma:shortregular}.

Recall, by Asm.~\ref{asm:nshift}, $T \subseteq T_{\leq \gamma} \cdot (\Sigma_\inp^*+\Sigma_\outp^*)^n$ and by Def.~\ref{def:M}, $T_i = T \cap \left( T_{\leq \gamma} \cdot (\Sigma_\inp^*+\Sigma_\outp^{\leq i})^n \right)$ for $i \geq 0$.

\begin{restatable}{lemma}{lemmashortregular}\label{lemma:shortregular}
 If $S$ has a $T$-controlled uniformizer, then $S$ has a $T_k$-controlled uniformizer for a computable $k \geq 0$.
\end{restatable}

The proof of the above lemma yields that $k$ can be chosen as $r_1 + r_2$.
This concludes the first part of the proof of Theorem~\ref{thm:regular}.
For the second part, we prove that the problem whether $S$ has a $T_i$-controlled uniformizer for an $i$ reduces to the question whether $T_i(S)$ has a subset uniformizer for a suitable $T_i(S)$ as defined below in Lemma~\ref{lemma:transformregular}.

\subparagraph*{Reduction.}

The next lemma shows that from $S$ a regular $T_i(S)$ can be obtained such that $T_i(S)$ consists of all $T_i$-controlled synchronizations $w$ with $\llbracket w \rrbracket \in \llbracket S \rrbracket$.

\begin{restatable}{lemma}{lemmatransformregular}\label{lemma:transformregular}
 For $i \geq 0$, the language $T_i(S) = \{ w \in \Sigma_{\inp\outp}^* \mid w \in T_i \text{ and }\allowbreak \llbracket w \rrbracket \in \llbracket S \rrbracket\}$ is a $T_i$-controlled effectively constructible regular language.
\end{restatable}

We are ready to prove the main theorem of this paper.

\begin{proof}[Proof sketch of Theorem~\ref{thm:regular}]
By Lemma \ref{lemma:shortregular} we know that if $S$ has a $T$-controlled uniformizer, then $S$ has a $T_k$-controlled uniformizer for a computable $k \geq 0$.
Let $T_k(S)$ be defined as in Lemma~\ref{lemma:transformregular}.

We can show that $S$ has a $T$-controlled uniformizer iff $\mathrm{dom}(\llbracket S \rrbracket)\allowbreak =\allowbreak \mathrm{dom}(\llbracket T_k(S) \rrbracket)$ and $T_k(S)$ has a subset uniformizer which is decidable by Theorem~\ref{thm:sunif}. 
\end{proof}

%%%%%%%%%%%%%%%%%%%%%%%%%%%%%%%%%%%%%%%%%%%%%%%%%%%%%%%%%%%%%%%%%%%%%%%%%%%
%%%%%%%%%%%%%%%%%%%%%%%%%%%%%%%%%%%%%%%%%%%%%%%%%%%%%%%%%%%%%%%%%%%%%%%%%%%
%%%%%%%%%%%%%%%%%%%%%%%%%%%%%%%%%%%%%%%%%%%%%%%%%%%%%%%%%%%%%%%%%%%%%%%%%%%
\section{Conclusion}\label{sec:conclusion}

In this paper we considered uniformization by subsequential transducers in which the allowed input/output behavior is specified by a regular set of synchronizations, the so-called resynchronized uniformization problem.
An overview over our results can be found in Table~\ref{tab:overview}. 
For future work we want to study other problems of this kind, e.g., study whether the resynchronized uniformization problem is decidable for a given rational relation as source language and a given `recognizable' target language in the sense that the target language is controlled by a synchronization language that synchronizes recognizable relations.

%%%%%%%%%%%%%%%%%%%%%%%%%%%%%%%%%%%%%%%%%%%%%%%%%%%%%%%%%%%%%%%%%%%%%%%%%%%
%%%%%%%%%%%%%%%%%%%%%%%%%%%%%%%%%%%%%%%%%%%%%%%%%%%%%%%%%%%%%%%%%%%%%%%%%%%
%%%%%%%%%%%%%%%%%%%%%%%%%%%%%%%%%%%%%%%%%%%%%%%%%%%%%%%%%%%%%%%%%%%%%%%%%%%

\subparagraph*{Acknowledgements.} The author would like to thank her supervisor Christof L{\"o}ding for suggesting this topic and his helpful comments and thank the anonymous reviewers of this and an earlier version of the paper for their feedback which greatly improved the presentation.

%%
%% Bibliography
%%

%% Please use bibtex, 

\bibliography{biblio}

\appendix

\newpage

\section*{Appendix}

This is the full version of \url{http://dx.doi.org/10.4230/LIPIcs.ICALP.2018.281}.

\section{Details of Section \ref{sec:unifproblems}}

\subsection{Uniformizations of recognizable relations}

Here we present the result stating that it is decidable whether a given recognizable relation has a uniformization by a subsequential transducer for any given synchronization parameter.

\begin{theorem}
 Given a regular source language with finite $\mathit{shift}$ and a regular target language.
 Then, the resynchronized uniformization problem is decidable.
\end{theorem}

Let $S \subseteq \Sigma_{\inp\outp}$ denote a regular source language with finite $\mathit{shift}$ and $T \subseteq \Sigma_{\inp\outp}$ a regular target language.
Note that the usual presentation of a regular relation as $\bigcup_{i=1}^n U_i \times V_i$, where each $U_i$ and $V_i$ are regular languages over $\Sigma_\inp$ and $\Sigma_\outp$, respectively, clearly is representable as a regular language over $\Sigma_{\inp\outp}$ with finite $\mathit{shift}$, namely as $\bigcup_{i=1}^n U_i \cdot V_i$.

In \cite{conf/stacs/FigueiraL14}, it is shown that $1^*2^*$ is an effective canonical representative of the class of regular languages with finite $\mathit{shift}$.

\begin{proof}
Let $S$ and $T$ be as above.
We show the theorem in two steps.
 
First, we effectively compute the regular language $T' = \{ w \mid w \in T \text{ and } \llbracket w \rrbracket \in \llbracket S \rrbracket\}$, that is, the language that contains every $T$-controlled word that describes a pair from $\llbracket S \rrbracket$.
 
Secondly, we show that $S$ has a $T$-controlled uniformization by an sDFA if, and only if, $\mathrm{dom}(\llbracket S \rrbracket) = \mathrm{dom}(\llbracket T' \rrbracket)$ and $T'$ has a subset uniformization by an sDFA, which is decidable by Theorem~\ref{thm:sunif}.
 
For the first part, let $\mathcal A$ be a DFA that recognizes the $1^*2^*$-controlled canonical representation of $S$.
Consider an NFA that on reading a word $w \in \Sigma_{\inp\outp}^*$ works as follows.
First, it guesses a state $q \in Q_\mathcal A$, then it simulates $\mathcal A$ on $\pi_\inp(w)$ from $q_0$ and $\mathcal A$ on $\pi_\outp(w)$ from $q$.
It accepts if $\delta_\mathcal A^*(q_0,\pi_\inp(w)) = q$ and $\delta_\mathcal A^*(q,\pi_\outp(w)) \in F_\mathcal A$.
The intersection of this language with $T$ is our desired language $T'$.
 
For the second part, assume $\mathrm{dom}(\llbracket S \rrbracket) = \mathrm{dom}(\llbracket T' \rrbracket)$ and $T'$ has a subset uniformization by an sDFA.
Since $\mathrm{dom}(\llbracket S \rrbracket) = \mathrm{dom}(\llbracket T' \rrbracket)$, every subset uniformization of $T'$ is also a $T$-controlled uniformization of $S$.

For the other direction, assume $S$ has a $T$-controlled uniformization by an sDFA, say $U$.
Obviously $\llbracket U \rrbracket \subseteq_\mathit{u} \llbracket S \rrbracket$ and $\mathrm{dom}(\llbracket S \rrbracket) = \mathrm{dom}(\llbracket U \rrbracket)$.
First, we show $\mathrm{dom}(\llbracket S \rrbracket) = \mathrm{dom}(\llbracket T' \rrbracket)$.
Proof by contradiction, assume there is some $u \in \mathrm{dom}(\llbracket S \rrbracket)\setminus \mathrm{dom}(\llbracket T' \rrbracket)$.
There exists a $T$-controlled $w \in U$ such that $\pi_\inp(w) = u$.
By construction, $w \in T'$, thus $u \in \mathrm{dom}(\llbracket T' \rrbracket)$.
Thus, $\mathrm{dom}(\llbracket S \rrbracket) = \mathrm{dom}(\llbracket T' \rrbracket) = \mathrm{dom}(\llbracket U \rrbracket)$.
Secondly, since $U \subseteq T$ and $\mathrm{dom}(\llbracket U \rrbracket) = \mathrm{dom}(\llbracket T' \rrbracket)$, it is clear that $U$ is a subset uniformization of $T' \subseteq T$.
 
\end{proof}

\subsection{Parikh-injective synchronization languages}

{
\renewcommand{\thetheorem}{\ref{thm:parikh}}
\thmparikh*
\addtocounter{theorem}{-1}
}

\begin{proof}[Proof of Proposition~\ref{thm:parikh}]
 We show that every $T$-controlled uniformization of $S$ is in fact a subset uniformization of $S$.
  
 Towards a contradiction, assume that $U$ is a $T$-controlled uniformization, but $U \not\subseteq S$.
 
 Since $U$ is $T$-controlled, $U$ is $L$-controlled.
 There is $w \in U \setminus S$ with $\llbracket w \rrbracket \in \llbracket S \rrbracket$ and $w' \in S \setminus U$ with $\llbracket w \rrbracket = \llbracket w' \rrbracket$.
 Let $w = u \otimes v$ and $w' = u' \otimes v'$.
 Since $\llbracket w \rrbracket = \llbracket w' \rrbracket$ and both $v$ and $v'$ are $L$-controlled, it follows that $\Pi_L(v) = \Pi_L(v')$.
 Assume $v \neq v'$, this is a contradiction because $L$ is Parikh-injective.
 Thus, $v = v'$ and $u \neq u'$, because $w \neq w'$.
 This is a contradiction to $\llbracket w \rrbracket = \llbracket w' \rrbracket$.
 Hence, $U \subseteq S$, i.e., $U$ is a subset uniformization. 
\end{proof}
 
\section{Details of Section \ref{sec:regular}}

\subsection{State transformation trees.}

Analogously, we define output state transformation trees, where the roles of input and output are reversed compared to input state transformation trees. 

\begin{definition}[Output state transformation tree]\label{def:outputstt}
Given $i \geq 0$, $p \in Q_\mathcal B$, $q \in Q_\mathcal A$ and an output word $y \in \Sigma_\outp^*$, the \emph{state transformation tree} $\mathrm{STT}^i(y,p,q)$ is a tree over $Q_\mathcal B \times Q_\mathcal A$ defined inductively.
 \begin{itemize}[topsep=0pt]
 \item 
For $i = 0$, the tree $\mathrm{STT}^0(y,p,q)$ is built up as follows.

\noindent Let $\mathrm{Reach}_0 \subseteq Q_\mathcal B \times Q_\mathcal A$ be the smallest set such that
$(p',q') \in \mathrm{Reach}_0$ if there is some $x \in \Sigma_\inp^*$ with $|x| \leq |y|$ such that $\delta_\mathcal A^*(q,(x,y)) = q'$ and $\delta_\mathcal B^*(p,x) = p'$.

Then $\mathrm{STT}^0(y,p,q) = (p,q)({r_1}\dots{r_n})$
% \begingroup
% \setlength{\abovedisplayskip}{.5\columnsep }
% \setlength{\belowdisplayskip}{.5\columnsep }
% \begin{equation*}
% (p,q)({r_1}\dots{r_n})
% \end{equation*}
% \endgroup
for $\mathrm{Reach}_0 = \{r_1,\dots,r_n\}$.

\item For $i > 0$, the tree $\mathrm{STT}^i(y,p,q)$ is built up as follows.

\noindent Let $\mathrm{Reach}_1 \subseteq \Sigma_\outp^* \times Q_\mathcal B \times Q_\mathcal A$ be the smallest set such that $(y'',p',q') \in \mathrm{Reach}_1$ if 
\begin{itemize}
 \item $y = y'y''$ with $y''\in \Sigma_\outp^+$ for a $y' \in \Sigma_\outp^+$ such that there is an $x \in \Sigma_\inp^+$ with $|x| = |y'|$, and
 \item  $\delta_\mathcal A^*(q,(x,y')) = q'$ and $\delta_\mathcal B^*(p,x) = p'$.
\end{itemize}
For $(y'',p',q') \in \mathrm{Reach}_1$, let $\mathrm{Reach}_{(y'',p',q')} \subseteq \Sigma_\inp^* \times Q_\mathcal B \times Q_\mathcal A$ be the smallest set such that $(y'',p'',q') \in \mathrm{Reach}_{(y'',p',q)}$ if $\delta_\mathcal B^*(p',w) = p''$ for some $w \in \Sigma_\outp^+$.
Furthermore, let the tree $t_{(y'',p',q ')}^{i-1}$ be defined as $(p',q')(\mathrm{STT}^{i-1}{r_1}\dots\mathrm{STT}^{i-1}{r_n})$ for $\mathrm{Reach}_{(y'',p',q')} = \{r_1,\dots,r_n\}$.

Then the tree $\mathrm{STT}^i(y,p,q)$ is defined as 
\begingroup
\setlength{\abovedisplayskip}{.5\columnsep }
\setlength{\belowdisplayskip}{.5\columnsep }
\begin{equation*}
\mathrm{STT}^0(y,p,q) \circ (p,q)(t_{s_1}^{i-1}\dots t_{s_n}^{i-1}) 
\end{equation*}
\endgroup
for $\mathrm{Reach}_1 = \{s_1,\dots,s_n\}$.
\end{itemize}
\end{definition}

Now that we have defined output state transformation trees, we need to introduce one more concept, before we can define output profiles. 

Ultimately, given a uniformizer, our goal is to replace large segments that cause lag with (short) segments that have the same profile.
Towards defining profiles for output words it turns out that we need to store additional information compared to input profiles.
Intuitively, a difference arises because waiting a long time before output is produced (i.e., causing large input lag) means that lots of information about the input is known before output is produced;
whereas producing large output segments (i.e., causing large output lag) means that output has been produced without prior knowledge of the input.
Therefore, we introduce the concept of annotated output state transformation trees which model the possible interactions between input segments and the given output segment in more detail compared to output state transformation trees.
More specifically, for an input segment $x$, we collect vertices that can be reached by prefixes of $x$.
Below a formal definition is given and in Ex.~\ref{ex:outputSTT} an intuitive example is given.

\begin{definition}[Annotated output state transformation tree]\label{def:anntree}
 Let $i \geq 0$, $p \in Q_\mathcal B$, $q \in Q_\mathcal A$, $y \in \Sigma_\outp^*$, and let $t = (V_t,E_t,v_t,\val{t})$ denote the reduced output state transformation tree $\mathit{red}\bigl(\mathrm{STT}^i(y,p,q)\bigr)$.

 For $v \in V_t$, the \emph{annotated output state transformation tree} $\mathrm{annSTT}^i(y,p,q,v)$ is a tree over $(Q_\mathcal B \times Q_\mathcal A \times {V_t}) \cup (Q_\mathcal B \times Q_\mathcal A \times V_t \times 2^{V_t})$ defined inductively.
 \begin{itemize}[topsep=0pt]
 \item 
 For $i = 0$, the tree $\mathrm{annSTT}^0(y,p,q,v)$ is built up as follows.

 \noindent Let $\mathrm{Reach}_0 \subseteq Q_\mathcal B \times Q_\mathcal A \times V_t \times 2^{V_t}$ be the smallest set such that $(p',q',v',S) \in \mathrm{Reach}_0$ if there is some $x \in \Sigma_\inp^*$ with $|x| \leq |y|$ such that
 \begin{itemize}
  \item $\delta_\mathcal A^*(q,(x,y)) = q'$ and $\delta_\mathcal B^*(p,x) = p'$, and
  \item $x$ leads from $v$ to $v'$ w.r.t.\ $y$ and $0$, and
  \item $v'' \in S$ if there is $x'\sqsubseteq x$ such that $x'$ leads from $v$ to $v''$ w.r.t.\ $y$ and $0$.
 \end{itemize}
 Then $\mathrm{annSTT}^0(y,p,q,v) = (p,q,v)({r_1}\dots{r_n})$ for $\mathrm{Reach}_0 = \{r_1,\dots,r_n\}$.

 \item For $i > 0$, the tree $\mathrm{annSTT}^i(y,p,q,v)$ is built up as follows.

 \noindent Let $\mathrm{Reach}_1 \subseteq \Sigma_\outp^* \times Q_\mathcal B \times Q_\mathcal A \times V_t \times 2^{V_t}$ be the smallest set such that $(y'',p',q',v',S) \in \mathrm{Reach}_1$ if there is some $x \in \Sigma_\inp^+$ with  $|x| < |y|$  such that
 \begin{itemize}
  \item $y = y'y''$ with $y''\in \Sigma_\outp^+$ for $y' \in \Sigma_\outp^+$ with $|x| = |y'|$, and
  \item $\delta_\mathcal A^*(q,(x,y')) = q'$ and $\delta_\mathcal B^*(p,x) = p'$, and
  \item $x$ leads from $v$ to $v'$ w.r.t.\ $y$ and $i$, and
  \item $v'' \in S$ if there is $x'\sqsubseteq x$ such that $x'$ leads from $v$ to $v''$ w.r.t.\ $y$ and $i$.
 \end{itemize}
 For $(y'',p',q',v',S) \in \mathrm{Reach}_1$, let $\mathrm{Reach}_{(y'',p',q',v',S)} \subseteq \Sigma_\outp^* \times Q_\mathcal B \times Q_\mathcal A \times V_t$ be the smallest set such that $(y'',p'',q',v'') \in \mathrm{Reach}_{(y'',p',q',v',S)}$ if $\delta_\mathcal B^*(p',w) = p''$ and $w$ leads from $v'$ to $v''$ for some $w \in \Sigma_\outp^+$.

 \noindent Furthermore, let $t_{(y'',p',q',v',S)}^{i-1}$ be the tree
 \begingroup
 \setlength{\abovedisplayskip}{.5\columnsep }
 \setlength{\belowdisplayskip}{.5\columnsep }
 \begin{equation*}
 (p',q',v',S)(\mathrm{annSTT}^{i-1}{r_1}\dots\mathrm{annSTT}^{i-1}{r_n})
 \end{equation*}
 \endgroup
 for $\mathrm{Reach}_{(y'',p',q',v',S)} = \{r_1,\dots,r_n\}$.
 
 Finally, the tree $\mathrm{annSTT}^i(y,p,q,v)$ is defined as
 \begingroup
 \setlength{\abovedisplayskip}{.5\columnsep }
 \setlength{\belowdisplayskip}{.5\columnsep }
 \begin{equation*}
  \mathrm{annSTT}^0(y,p,q,v) \circ (p,q,v)(t_{s_1}^{i-1}\dots t_{s_n}^{i-1})
 \end{equation*}
 \endgroup
 for $\mathrm{Reach}_1 = \{s_1,\dots,s_n\}$.
 \end{itemize}
\end{definition}

Next, we give an example to illustrate the difference between output STTs and annotated output STTs.

\begin{example}\label{ex:outputSTT}
 Given an alphabet $\Sigma_{\inp\outp}$ with $\Sigma_\inp = \{a,b\}$ and $\Sigma_\outp = \{c\}$, an automatic relation $S_1$ over $\Sigma_{\inp\outp}$ is given by a DFA $\mathcal A_1$ depicted in Fig.~\ref{fig:dfaA}, and an automatic relation $T_1$ over $\Sigma_{\inp\outp}$  is given by a DFA $\mathcal B_1$ depicted in Fig.~\ref{fig:dfaB}.
 Note that, $S_1$ is $(12)^*(1^*+2^*)$-controlled, i.e., canonical, hence, the notion of state transformation tree is meaningful w.r.t.\ $\mathcal A_1$ and $\mathcal B_1$.
 
 Consider the output word $cc \in \Sigma_\outp^*$, the reduced variant of the output state transformation tree $\mathrm{STT}^0(cc,p_0,q_0)$ is depicted in Fig.~\ref{fig:stt}. Additionally, its edges are labeled with the respective associated words. %(w.r.t.\ $cc$ and $0$).
 Also, its vertices are named, so that they can be referred to in the annotated output state transformation tree $\mathrm{annSTT}^0(cc,p_0,q_0)$ depicted in Fig.~\ref{fig:annstt}.
 
 Compared to $\mathit{red}(\mathrm{STT}^0(cc,p_0,q_0))$ we can see that $v_3$ was duplicated with annotation $(v_3,\{v_1,v_3\})$ and $(v_3,\{v_2,v_3\})$, respectively.
 This has happened because both $ab$ and $ba$ lead from $v_0$ to $v_3$, but $a$ (prefix of $ab$) leads from $v_0$ to $v_1$ and $b$ (prefix of $ba$) leads from $v_0$ to $v_2$.
\end{example}

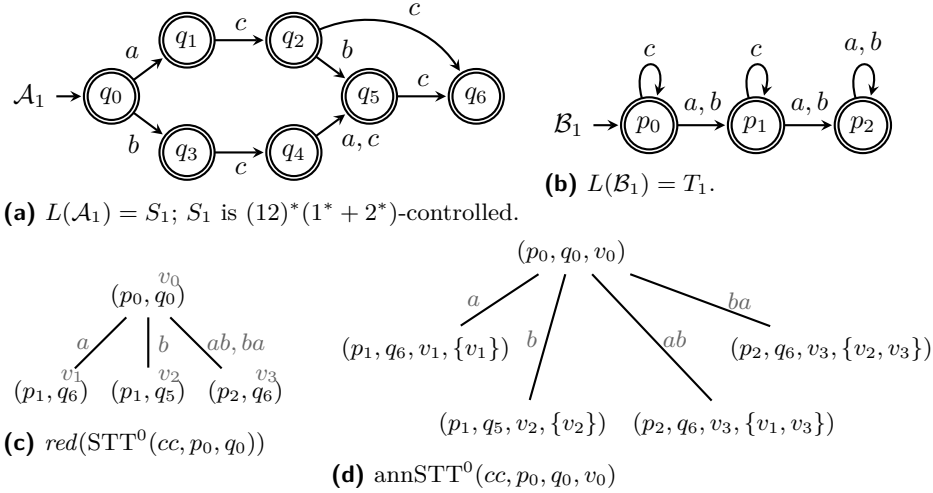
\begin{figure}
\vskip -0.5cm
\centering

 \begin{subfigure}{0.5\textwidth}
  \begin{tikzpicture}[scale=0.8,thick]
    \tikzstyle{every state}+=[inner sep=3pt, minimum size=3pt];
    \node[state, initial, accepting, initial text=$\mathcal A_1$] (0) {$q_0$};
    \node[state, accepting, above right of=0, yshift=-0.25cm] (1) {$q_1$};
    \node[state, accepting, right of=1]       (2) {$q_2$};
    \node[state, accepting, below right of=0, yshift=0.25cm] (3) {$q_3$};
    \node[state, accepting, right of=3]       (4) {$q_4$};
    \node[state, accepting, below right of=2, yshift=0.25cm] (5) {$q_5$};
    \node[state, accepting, right of=5]       (6) {$q_6$};
    \draw[->] (0) edge                node        {$a$} (1);
    \draw[->] (0) edge                node[swap]  {$b$} (3);
    \draw[->] (1) edge                node        {$c$} (2);
    \draw[->] (2) edge                node        {$b$} (5);
    \draw[->] (2) edge[bend left=40]  node        {$c$} (6); 
    \draw[->] (3) edge                node[swap]  {$c$} (4);
    \draw[->] (4) edge                node[swap]  {$a,c$} (5);
    \draw[->] (5) edge                node        {$c$} (6);
   \end{tikzpicture}
   
  \caption{$L(\mathcal A_1) = S_1$; $S_1$ is $(12)^*(1^*+2^*)$-controlled.\label{fig:dfaA}} 
  \end{subfigure}
  \begin{subfigure}{0.3\textwidth}
  \begin{tikzpicture}[thick,scale=0.8]
    \tikzstyle{every state}+=[inner sep=3pt, minimum size=3pt];

    \node[state, initial, accepting, initial text=$\mathcal B_1$] (0) {$p_0$};
    \node[state, accepting, right of=0] (1) {$p_1$};
    \node[state, accepting, right of=1] (2) {$p_2$};
    \draw[->] (0) edge[loop above]    node    {$c$} ();
    \draw[->] (0) edge                node    {$a,b$} (1);
    \draw[->] (1) edge[loop above]    node    {$c$} ();
    \draw[->] (1) edge                node    {$a,b$} (2);
    \draw[->] (2) edge[loop above]    node    {$a,b$} ();
   \end{tikzpicture}
  \caption{$L(\mathcal B_1) = T_1$.\label{fig:dfaB}} 
  \end{subfigure}

  \begin{subfigure}{0.3\textwidth}
   \begin{tikzpicture}[thick,scale=0.8,baseline=(current bounding box.base)]
      \tikzstyle{level 1}=[sibling distance=16mm]
      \path[level distance=16mm] node (root){\small$(p_0,q_0)$}
	child{
	  node(0){\small$(p_1,q_6)$}	   
	}
	child{
	  node(1){\small$(p_1,q_5)$}
	}
	child{
	  node(2){\small$(p_2,q_6)$}
	}
      ;
      
      \path (root) -- coordinate[midway] (r0) (0);
      \path (root) -- coordinate[midway] (r1) (1);
      \path (root) -- coordinate[midway] (r2) (2);
      
      \node [dark-gray,left] at (r0) {\small$a$};
      \node [dark-gray,right] at (r1) {\small$b$};
      \node [dark-gray,right] at (r2) {\small$ab,ba$};
      
      \node [dark-gray,above right] at (root) {\small$v_0$};
      \node [dark-gray,above right] at (0) {\small$v_1$};
      \node [dark-gray,above right] at (1) {\small$v_2$};
      \node [dark-gray,above right] at (2) {\small$v_3$};
     \end{tikzpicture}
     
  \caption{$\mathit{red}(\mathrm{STT}^0(cc,p_0,q_0))$\label{fig:stt}} 
  \end{subfigure}
  \begin{subfigure}{0.5\textwidth}
    \begin{tikzpicture}[thick,scale=0.8,baseline=(current bounding box.base)]
      \tikzstyle{level 1}=[sibling distance=16mm]
      \path[level distance=16mm] node (root){\small$(p_0,q_0,v_0)$}
	child{
	  node(0){\small$(p_1,q_6,v_1,\{v_1\})$}	   
	}
	child{
	  node[yshift=-1cm](1){\small$(p_1,q_5,v_2,\{v_2\})$}
	}
	child{
	  node[yshift=-1cm,xshift=1.5cm](2){\small$(p_2,q_6,v_3,\{v_1,v_3\})$}
	}
	child{
	  node[xshift=1.5cm](3){\small$(p_2,q_6,v_3,\{v_2,v_3\})$}
	}
      ;
      
      \path (root) -- coordinate[midway] (r0) (0);
      \path (root) -- coordinate[midway] (r1) (1);
      \path (root) -- coordinate[midway] (r2) (2);
      \path (root) -- coordinate[midway] (r3) (3);
      
      \node [dark-gray,left] at (r0) {\small$a\ $};
      \node [dark-gray,left] at (r1) {\small$b$};
      \node [dark-gray,right] at (r2) {\small$ab$};
      \node [dark-gray,right] at (r3) {\small$\ \ ba$};
     \end{tikzpicture}
     
  \caption{$\mathrm{annSTT}^0(cc,p_0,q_0,v_0)$\label{fig:annstt}} 
  \end{subfigure}

 \caption{Let $\Sigma_\inp = \{a,b\}$ and $\Sigma_\outp = \{c\}$. 
   Reduced variant of the output state transformation tree $\mathrm{STT}^0(cc,p_0,q_0)$ and the corresponding annotated tree $\mathrm{annSTT}^0(cc,p_0,q_0,v_0)$ both w.r.t.\ $\mathcal A_1$ and $\mathcal B_1$.
   The edges are labeled with its respective associated words. %w.r.t.\ $cc$ and $0$.
   See Ex.~\ref{ex:outputSTT} for a comparison of the trees. 
 }
 \label{fig:outputSTT}
\vskip -0.5cm
\end{figure}

We are ready to define profiles based on state transformation trees, but beforehand we introduce some terminology to speak more conveniently about state transformation trees.

We now formally define the concept of associated words.
Examples can be found in Fig.\ref{subfig:tree}, Fig.~\ref{fig:stt}, and Fig.~\ref{fig:annstt}.

\begin{definition}[Associated words]\label{def:associatedwords}
 Let $t = (V_t,E_t,v_t,\val{t})$  be an input STT.
 
 Given $v,v' \in V_t$ such that $(v,v') \in E_t$ and $v$ is on an even level, let $\val{t}(v) = (p,q)$ and $\val{t}(v') = (p',q')$.
 We say that \emph{$y \in \Sigma_\outp^*$ leads from $v$ to $v'$} w.r.t.\ $x$ and $i$ if $t|_v = \mathrm{STT}^i(x,p,q)$ and there is $x',x'' \in \Sigma_\inp^*$ such that $x=x'x''$, and $\delta_\mathcal A^*(q,(x',y)) = q'$, and $\delta_\mathcal B^*(p,y) = p'$, and $\{ t|_{v''} \mid v'' \in \mathit{children}_t(v')\} = \{ \mathrm{STT}^{i-1}(x'',p'',q') \mid \delta_\mathcal B^*(p',w) = p'' \text{ for some } w \in \Sigma_\inp^+\}$.
 
 Given  $v',v'' \in V_t$ such that $(v',v'') \in E_t$ and $v'$ is on an odd level, let $\val{t}(v') = (p',q')$ and $\val{t}(v'') = (p'',q')$.
 We say that \emph{$w \in \Sigma_\inp^+$ leads from $v'$ to $v''$} if $\delta_\mathcal B^*(p',w) = p''$.

 Analogously, we define these properties for output and annotated output STTs.
\end{definition}

For convenience, we introduce the following definition.

\begin{definition}[ann]
  Let $t_{\mathit{ann}}$ be an annotated output state transformation tree based on the reduced state transformation tree $t$.
  We define a function $\mathit{ann}\colon V_{t_{\mathit{ann}}} \to V_t$ with $\mathit{ann}(v) = u$ if the third component of $v$s label is $u$.
\end{definition}

We state some simple observations about annotated output state transformation trees used in the upcoming proofs.

\begin{lemma}\label{lemma:basic}
 Let $t_{\mathit{ann}}$ be an annotated output state transformation tree based on the reduced state transformation tree $t$.
 \begin{enumerate}
  \item If $x \in \Sigma_\inp^*$ leads from $v$ to $v'$ w.r.t.\ $z \in \Sigma_\outp^*$ and $i \geq 0$ in $t_{\mathit{ann}}$, then $x$ leads from $\mathit{ann}(v)$ to $\mathit{ann}(v')$ w.r.t.\ $z$ and $i$ in $t$.
  \item If $y \in \Sigma_\outp^*$ leads from $v$ to $v'$ in $t_{\mathit{ann}}$, then $y$ leads from $\mathit{ann}(v)$ to $\mathit{ann}(v')$ in $t$.
  \item Given $v \in V_{t_{\mathit{ann}}}$, if the first two components of its label are $(p,q)$, then $\mathit{ann}(v) \in V_t$ is labeled $(p,q)$.
 \end{enumerate}
\end{lemma}

\subsection{Profiles.}

We previously defined input profiles, now we define output profiles.

\begin{definition}[Output profile]
Given $y \in \Sigma_\outp^*$, we define its \emph{profile} $P_y$ as $(\tau_y,\mathrm{annSTT}_y^{\lceil n/2 \rceil})$, where $\mathrm{annSTT}_y^{\lceil n/2 \rceil} =$
 \begingroup
 \setlength{\abovedisplayskip}{.5\columnsep }
 \setlength{\belowdisplayskip}{.5\columnsep }
 \begin{equation*}
 \bigcup_{(p,q) \in Q_\mathcal B \times Q_\mathcal A}\!\!\! \{\mathit{red}\bigl(\mathrm{annSTT}^{\lceil n/2 \rceil}(y,p,q,v)\bigr) \mid v = \mathit{root}\bigl({\mathit{red}\bigl(\mathrm{STT}^{\lceil n/2 \rceil}(y,p,q)\bigr)}\bigr)\}.
 \end{equation*}
 \endgroup 
\end{definition}

We prove some properties of profiles.

{
\renewcommand{\thetheorem}{\ref{lemma:monoid}}
\lemmamonoid*
\addtocounter{theorem}{-1}
}

\begin{proof}[Proof of Lemma~\ref{lemma:monoid}]
 Given $x_1, x_2 \in \Sigma_\inp^+$, we show that the profile $P_{x_1x_2}$ of $x_1x_2 \in \Sigma_\inp^*$ can be computed from $P_{x_1}$ and $P_{x_2}$.
 
 The state transformation function $\tau_{x_1x_2}$ is defined as concatenation of the functions $\tau_{x_1}$ and $\tau_{x_2}$, i.e., $\tau_{x_1x_2}(p) = \tau_{x_2}(\tau_{x_1}(p))$.
 
 Recall Fig.~\ref{fig:STT} for an easier understanding of the following.
 
 Let $m = \lceil n/2 \rceil$, in order to compute the set $\mathrm{STT}_{x_1x_2}^{m}$ from $\mathrm{STT}_{x_1}^{m}$ and $\mathrm{STT}_{x_2}^{m}$, we need to make an observation first.
 For any $x \in \Sigma_\inp^*$, $p \in Q_\mathcal B$, $q \in Q_\mathcal A$, and $i \leq m$, the tree $\mathrm{STT}^i(x,p,q)$ can be obtained from the tree $\mathrm{STT}^m(x,p,q)$ by removing all non-trivial subtrees rooted at a vertex with height $2i+1$.
 Here, non-trivial is used to describe subtrees with more than one vertex, meaning leaves at height $2i+1$ are not removed.
 The same observation holds for its reduced variant, in the following we mean by tree always its reduced variant.
 
 For any $p \in Q_\mathcal B$ and $q \in Q_\mathcal A$, the tree $\mathrm{STT}^m(x_1x_2,p,q)$ can be obtained from the tree $\mathrm{STT}^m(x_1,p,q) = (V,E,v_0,\val{})$ by performing the following action for each pair $(u,v) \in E$ such that $v$ is a leaf.
 Let $v$ be at height $2i+1$ for some $i \geq 0$, note that leaves only occur at odd heights, let $\val{}(u) = (p_i',q_i)$ and $\val{}(v) = (p_{i+1},q_{i+1})$, and let $j$ be such that $m = (i+1)+j$.
 %We add new children to $u$ and to $v$.
 If $i$ is $m$, then we remove $v$, because this indicates that already $m+1$ output segments were used to consume $x_1$, however, at most $m+1$ output segments may be used to consume $x_1x_2$.
 Otherwise, we add new children to $v$.
 For all $p_{i+1}' \in Q_\mathcal B$ such that there is $w \in \Sigma_\inp^+$ with $\delta_\mathcal B^*(p_{i+1},w) = p_{i+1}'$ we add $\mathrm{STT}^j(x_2,p_{i+1}',q_{i+1})$ as a subtree to $v$.
 In any case, we add new children to $u$.
 Consider the tree $\mathrm{STT}^{j+1}(x_2,p_{i+1},q_{i+1})$, let it be of the form $(p_{i+1},q_{i+1})(t_1\dots t_k)$, then we add $t_1, \dots, t_k$ as subtrees to $u$.
 The meaning of this operation is to extend the output segment that leads from $u$ to $v$ beyond consuming (the remainder of) $x_1$ and also consuming parts of $x_2$.
 
 Hence, we are able to define a natural concatenation operation between input profiles. 
 Given $x_1, x_2 \in \Sigma_\inp^*$, let $P_{x_1}P_{x_2} = P_{x_1x_2}$.
 Thus, the set of input profiles is equipped with a concatenation operation and a neutral element, that is, the profile of the empty word, i.e., the set of input profiles is a monoid with concatenation.
 
 Given $y_1, y_2 \in \Sigma_\outp^+$, the profile $P_{y_1y_2}$ of $y_1y_2$ can be computed from $P_{y_1}$ and $P_{y_2}$ in the same way as described above for input profiles.
 This allows us to define a concatenation operation for output profiles as for input profiles, consequently, the set of output profiles is a monoid with concatenation. 
\end{proof}

{
\renewcommand{\thetheorem}{\ref{lemma:ramsey}}
\lemmaramsey*
\addtocounter{theorem}{-1}
}

\begin{proof}[Proof of Lemma~\ref{lemma:ramsey}]
 Ramsey's Theorem yields that for any number of colors $c$ and any number $r$, there exists a number $K \in \mathbbm{N}$ such that if the edges of a complete graph with at least $K$ vertices are colored with $c$ colors, then the graph must contain a complete subgraph with $r$ vertices such that all edges have the same color, see e.g.~\cite{diestel2000graduate}.
 
 Let $x \in \Sigma_\inp^*$ with the factorization $x = x_1x_2\dots x_n$, with $x_1,\dots,x_n \in 1 \times \Sigma$.
 Consider the complete graph $G = (V,E,\mathit{col})$ with edge-coloring $\mathit{col}: E \rightarrow \mathit{Cols}$, where $V := \{1,\dots,n\}$, $E := V \times V$, $\mathit{Cols}$ is the finite set of profiles and $\mathit{col}(e) := P_{x[i,k]}$ if $e = (i,k)$ for all $e \in E$.
 If there exist $i < j < k \leq n$ such that the edges $(i,j)$, $(j,k)$ and $(i,k)$ have the same color, i.e., the respective profiles are the same, then $x$ has a factorization that contains a non-empty idempotent factor.
 
 As a consequence of Ramsey's Theorem, if $|x|$ is equal or larger than the Ramsey number $R(3,|\mathit{Cols}|)$, then $x$ contains a non-empty idempotent factor. 
\end{proof}

\subsection{Proof of Theorem~\ref{thm:regular}.}

Recall, the proof of Theorem~\ref{thm:regular} is split in two parts.

\subsection*{Part I}

The goal is to show that if $S$ has a $T$-controlled uniformizer, then $S$ has a $T_k$-controlled uniformizer for a computable $k$; this is the statement of Lemma~\ref{lemma:shortregular}.

We introduce the following terminology used in the proofs of Lemmata~\ref{lemma:longeroutput}~and~\ref{lemma:shortregular}.

\begin{definition}
We say that \emph{$y$ fully traverses $x$ in $\mathcal A_q$} if $|y| \leq |x|$ and $\delta_\mathcal A^*(q,(x,y))$ is not a sink state, respectively, we say that \emph{$x$ fully traverses $y$ in $\mathcal A_q$} if $|x| \leq |y|$ and $\delta_\mathcal A^*(q,(x,y))$ is not a sink state.
Situations where $x$ and $y$ are of different length and $\delta_\mathcal A^*(q,(x,y))$ does not lead to a sink state can occur when $x \in \mathrm{dom}(L(\mathcal A_q))$.
\end{definition}

Before we can prove Lemma~\ref{lemma:shortregular}., we need to prove two auxiliary lemmata, namely Lemmata~\ref{lemma:longeroutput}~and~\ref{lemma:pumping}.

First, we prove Lemma~\ref{lemma:longeroutput} stating that there exists a bound $b$ such that it suffices to consider uniformizer where each output block increases the amount that the output sequence is ahead by at most $b$.

\begin{restatable}{lemma}{lemmalongeroutput}\label{lemma:longeroutput}
 There is a computable $b \geq 0$ such that if $S$ has a $T$-controlled uniformization by an sDFA, then $S$ has a $T$-controlled uniformization $U$ by an sDFA which satisfies the following property for each $w \in U$.

 For each $i < j$ such that $w[i]$ and $w[j]$ are consecutive shifts and $w[i+1] \in \Sigma_\outp$ it holds that
 $|\pi_\outp(w[1,i])| - |\pi_\inp(w[1,i])| \leq \mathit{max}\{|\pi_\outp(w[1,j])| - |\pi_\inp(w[1,j])|,0\} + b$.
\end{restatable}

% {
% \renewcommand{\thetheorem}{\ref{lemma:longeroutput}}
% \lemmalongeroutput*
% \addtocounter{theorem}{-1}
% }

\begin{proof}[Proof of Lemma~\ref{lemma:longeroutput}]
 Let $\beta$ be the smallest bound on the length of representatives of output profiles, then we chose $b$ to be $\mathit{max}\{\beta,\gamma+1\}$.

 Assume $U$ is a $T$-controlled uniformization given by a sequential DFA $\mathcal U$ that does not satisfy the property stated in the lemma.
 Recall, $T \subseteq T_{\leq \gamma} \cdot (\Sigma_\inp^*+\Sigma_\outp^*)^n$.
 If the bound on the length of output blocks stated in the lemma is violated, then we are in a situation where the lag has exceeded $\gamma$, thus it can be violated at most $\lceil n/2 \rceil$ times, because after the lag has exceeded $\gamma$ there are at most $n$ shifts, i.e., at most $\lceil n/2 \rceil$ output blocks.
 Let $m = \lceil n/2 \rceil$.
 
 We construct a $T$-controlled uniformization $U'$ recognized by a sequential DFA $\mathcal U'$ based on $\mathcal U$ that repairs for every input word the first violation of the output block length. Applying the construction presented below at most $m$ times yields a uniformization according to the statement of the lemma.
 
 The computation of $\mathcal U'$ differs from $\mathcal U$ from the point on that the following situation occurs: 
 Consider an arbitrary $w \in \mathit{Pref}(U)$ of length $\ell_2$, such that there is a position $\ell_1 < \ell_2$ such that $w[\ell_1]$ and $w[\ell_2]$ are consecutive shifts, $w[\ell_1+1] \in \Sigma_\outp$ and it holds that
 $|\pi_\outp(w[1,\ell_1])| - |\pi_\inp(w[1,\ell_1])| > \mathit{max}\{|\pi_\outp(w[1,\ell_2])| - |\pi_\inp(w[1,\ell_2])|,0\} + b$, i.e., the output block $w[\ell_1+1,\ell_2]$ has increased the lag caused by output symbols being ahead by more than $b$.
 Let $\ell$ be the smallest position $\ell_1 < \ell \leq \ell_2$ such that $|\pi_\outp(w[1,\ell])| - |\pi_\inp(w[1,\ell])| > 0$, that is, $\ell$ is the position such that $w[\ell,\ell_2]$ is the greatest part of the block $w[\ell_1+1,\ell_2]$ that is ahead of the input.
 Note that, the violation is caused because $w[\ell,\ell_2] > b$.

 We prove that we can replace the output $w[\ell,\ell_2]$ by some output of length at most $b$ chosen as follows.
 Let $w[\ell,\ell_2]$ be $y \in \Sigma_\outp^+$, consider the profile $P_y$ and let $z$ be a representative of $P_y$.
 We show that we can replace $y$ by $z$.
 Let $w$ have the factorization $xy$.
 Since $P_y = P_z$, we have $\tau_y = \tau_z$, and $\mathrm{annSTT}^m_y = \mathrm{annSTT}^m_z$.
 Let $\delta_\mathcal A^*(q_\mathcal A^0,(\pi_\inp(x),\pi_\outp(x))) = q$, $\delta_\mathcal B^*(q_\mathcal B^0,(x) = p$, and let 
 \begin{itemize}
  \item $t_y$ denote $\mathit{red}(\mathrm{STT}^{\lceil n/2 \rceil}(y,p,q))$, and
  \item $t_z$ denote $\mathit{red}(\mathrm{STT}^{\lceil n/2 \rceil}(z,p,q))$, and
  \item $t^{\mathit{ann}}_y$ denote $\mathrm{STT}^{\lceil n/2 \rceil}(z,p,q,\ro{t_y})$, and
  \item $t^{\mathit{ann}}_z$ denote $\mathrm{STT}^{\lceil n/2 \rceil}(z,p,q,\ro{t_z})$.
  \end{itemize}
 Clearly, $t_y = t_z$ and $t^{\mathit{ann}}_y = t^{\mathit{ann}}_z$.
 
 Assume we already have defined $\mathcal U'$ up to the point where the violation as stated above occurs, and until then, $\mathcal U$ and $\mathcal U'$ have worked exactly the same way.
 We show that $\mathcal U'$ can continue the computation successfully after replacing $y$ with $z$ by showing that there exists a sequentially computable run satisfying the following properties.
 
 Let $\delta_\mathcal U^*(q_\mathcal U^0,xy) = s$, and $\delta_{\mathcal U'}^*(q_{\mathcal U'}^0,xz) = r$, for the next input symbols until $z$ is fully traversed in $\mathcal A_q$, we inductively (on the number of input blocks) define the computation of $\mathcal U'_r$ satisfying the following properties:
 
 For $w_1,\dots, w_i \in \Sigma_\inp^+$ with $|w_1\dots w_i| < |z|$ and $o_1,\dots, o_i \in \Sigma_\outp^+$ such that $w_1o_1\dots w_io_i \in \mathit{Pref}(U'_{r})$, there exists a path $v_0'v_1v_1'\dots v_iv_i'$ in $t^{\mathit{ann}}_z$ with $v_0' = \ro{t^{\mathit{ann}}_z}$ such that 
 \begin{enumerate}
  \item $w_j$ leads from $v_{j-1}'$ to $v_j$ in $t^{\mathit{ann}}_z$ for all $1 \leq j \leq i$, and
  \item $o_j$ leads from $v_j$ to $v_{j}'$ in $t^{\mathit{ann}}_z$ for all $1 \leq j \leq i$, and
  \item there exist $\bar w_1,\dots \bar w_i \in \Sigma_\inp^+$ and a path $u_0'u_1u_1'\dots u_iu_i'$ in $t_y$ with $u_0' = \ro{t_y}$ such that
  \begin{enumerate}
   \item $\mathit{ann}(v_j) = u_j$ and $\mathit{ann}(v_j') = u_j'$ for all $0 \leq j \leq i$, and 
   \item $\bar w_j$ leads from $u_{j-1}'$ to $u_j$ in $t_y$ for all $1 \leq j \leq i$, and
   \item $o_j$ leads from $u_j$ to $u_{j}'$  in $t_y$ for all $1 \leq j \leq i$, and
   \item $\bar w_1o_1\dots \bar w_io_i \in \mathit{Pref}(U_{s})$.
  \end{enumerate}
 \end{enumerate}
 \noindent Additionally, for $w_{i+1} \in \Sigma_\inp^*$ with $|w_1\dots w_{i+1}| \leq |z|$ such that $w_1o_1\dots w_io_iw_{i+1} \in \mathit{Pref}(U'_{r})$ and $w_1\dots w_{i+1}$ fully traverses $z$ in $\mathcal A_q$, there exists a leaf node $v_{i+1}$ in $t^{\mathit{ann}}_z$ with $(v_{i}',v_{i+1}) \in E_{t^{\mathit{ann}}_z}$ such that 
 \begin{enumerate}
  \setcounter{enumi}{3}
  \item $w_{i+1}$ leads from $v_{i}'$ to $v_{i+1}$ in $t^{\mathit{ann}}_z$, and
  \item there exists $\bar w_{i+1} \in \Sigma_\inp^*$ and a leaf node $u_{i+1}$ in $t_y$ with $(u_{i}',u_{i+1}) \in E_{t_y}$ such that
  \begin{enumerate}
   \item $\mathit{ann}(v_{i+1}) = u_{i+1}$
   \item $\bar w_{i+1}$ leads from $u_{i}'$ to $u_{i+1}$ in $t_y$, and
   \item $\bar w_1o_1\dots \bar w_io_i\bar w_{i+1} \in \mathit{Pref}(U_{s})$.
  \end{enumerate}
 \end{enumerate}
 To be clear, the formulation $w_j$ leads from $v'_{j-1}$ to $v_j$ in $t_z^{\mathit{ann}}$ is used to mean that $w_j$ leads from $v'_{j-1}$ to $v_j$ w.r.t.\ $z''$ and $j'$, where $j' = \lceil n/2 \rceil - (j-1)$ and $z'' \in \Sigma_\outp^*$ such that $z$ has a factorization $z'z''$ with $|z'| = |w_1\dots w_{j-1}|$.
 Analogously, the formulation $\bar w_j$ leads from $u'_{j-1}$ to $u_j$ in $t_y$ is used to mean that $\bar w_j$ leads from $u'_{j-1}$ to $u_j$ w.r.t.\ $y''$ and $j'$, where $j' = \lceil n/2 \rceil - (j-1)$ and $y'' \in \Sigma_\outp^*$ such that $y$ has a factorization $y'y''$ with $|y'| = |\bar w_1\dots \bar w_{j-1}|$.
 
 Assume we have already defined the computation for some $k \leq i$ satisfying conditions $1$.--$3$., i.e., for $w_1,\dots, w_k \in \Sigma_\inp^+$ with $|w_1\dots w_k| < |z|$ and $o_1,\dots, o_k \in \Sigma_\outp^+$ such that the computation yields $w_1o_1\dots w_ko_k \in \mathit{Pref}(U'_{r})$, there exists a path $v_0'\dots v_k'$ in $t^{\mathit{ann}}_z$ with $v_0' = \ro{t^{\mathit{ann}}_z}$ such that $w_1o_1\dots w_ko_k$ leads from $v_0'$ to $v_k'$ in $t^{\mathit{ann}}_z$, there are $\bar w_1, \dots, \bar w_k \in \Sigma_\inp^+$ and a path $u_0'\dots u_k'$ in $t_y$ with $u_0' = \ro{t_y}$ such that $\mathit{ann}(v_0')\dots \mathit{ann}(v_k') = v_0'\dots v_k'$, $\bar w_1o_1\dots \bar w_ko_k$ leads from $u_0'$ to $u_k'$ in $t_y$, and $\bar w_1o_1\dots \bar w_ko_k \in \mathit{Pref}(U_s)$.
 
 To determine which part of the next (up to) $|z|-|w_1\dots w_k|$ input symbols will be $w_{k+1}$, we do the following after each read input symbol:
 Assume that after the $m$th input symbol the sequence $a_1\dots a_{m}$ has been read and let $a_1\dots a_m$ lead from $v_k'$ to $v$ in $t^{\mathit{ann}}_z$.
 Let $\mathit{ann}(v) = u$, note that, by construction, as stated in Lemma \ref{lemma:basic}, $a_1\dots a_m$ lead from $u_k'$ to $u$ in $t_z$.
  
 We distinguish two cases.
 \begin{description}
  \item $v$ is not a leaf, i.e., $w_1\dots w_{k}a_1\dots a_m$ does not fully traverse $z$ in $\mathcal A_q$.
  If there exists a $w' \in \Sigma_\inp^+$ such that $w'$ leads from $u_k'$ to $u$ in $t_y$ and there is an $o \in \Sigma_\outp^+$ such that $w'_1o_1\dots w'_ko_kw'o \in \mathit{Pref}(U_{s})$, then let $w_{k+1} = a_1\dots a_m$, $w_{k+1}' = w'$, $o_{k+1} = o$, $v_{k+1} = v$, $u_{k+1} = u$ and $w'_1o_1\dots w'_ko_kw_{k+1}o_{k+1} \in \mathit{Pref}(U'_{r})$.
  Meaning, $\mathcal U'_r$ produces output $o_{k+1}$ after reading $w_1\dots w_{k+1}$.
  Let $v_{k+1}'$ be the node such that $o_{k+1}$ leads from $v_{k+1}$ to this node in $t^{\mathit{ann}}_z$, and let $u_{k+1}' = \mathit{ann}(v_{k+1}')$.
  By Lemma \ref{lemma:basic}, $o_{k+1}$ also leads from $u_{k+1}$ to $u_{k+1}'$ in $t_y$.
  It is easy to see that conditions $1$.--$3$.\ are satisfied.
 
  Otherwise, if there exists no $w' \in \Sigma_\inp^+$ such that $w'$ leads from $u_k'$ to $u$ in $t_y$ and there is an $o \in \Sigma_\outp^+$ such that $w'_1o_1\dots w'_ko_kw'o \in \mathit{Pref}(U_{s})$, then we additionally consider the next input symbol.
  
  \item $v$ is a leaf, i.e., $w_1\dots w_{k}a_1\dots a_m$ fully traverses $z$ in $\mathcal A_q$.
  Then $k = i$, and let $w_{i+1} = a_1\dots a_m$, $v_{i+1} = v$ and $u_{i+1} = u$.
  Since $w_1o_1\dots w_io_iw_{i+1} \in \mathit{Pref}(U'_{r})$, we show that conditions $4$.--$5$.\ can be satisfied.
  Clearly, by construction of STTs, $u_{i+1} = u$ is also a leaf and condition $4$.\ is satisfied.
  Towards a contradiction, assume condition $5$.\ can not be satisfied, meaning for all $w' \in \Sigma_\inp^*$ such that $w'$ leads from $u_{i}'$ to $u_{i+1}$ it holds $w'_1o_1\dots w'_io_iw' \notin \mathit{Pref}(U_{s})$.
  Since $w'_1o_1\dots w'_io_i \in \mathit{Pref}(U_{s})$, for each such $w'$ there exists a factorization $x_1x_2$ such that $w'_1o_1\dots w'_io_ix_1 \in \mathit{Pref}(U_{s})$.
  Recall, $t^{\mathit{ann}}_z = t^{\mathit{ann}}_y$, thus there exists at least one $w'$ such that $w'$ leads from $v_i'$ to $v_{i+1}$ in $t^{\mathit{ann}}_y$ and a factorization of $w'$ to $x_1x_2$ such that $w'_1o_1\dots w'_io_ix_1 \in \mathit{Pref}(U_{s})$.
  This means there exists some node $\tilde v$ in $t^{\mathit{ann}}_y$ such that $x_1$ leads from $v_{i}'$ to $\tilde v$ in $t^{\mathit{ann}}_y$ and $x_1$ leads from $u_{i}'$ to $\mathit{ann}(\tilde v)$ in $t_y$.
  Let $\mathit{ann}(\tilde v) = \tilde u$ and let $\val{t^{\mathit{ann}}_z}(v_{i+1}) = (p_{i+1},q_{i+1},u_{i+1},V_{i+1})$, then by construction of annotated STTs, we obtain $\tilde u \in V_{i+1}$.
  Therefore, we know that since $w_{i+1}$ leads from $v_i'$ to $v_{i+1}$ in $t^{\mathit{ann}}_z$ there exists a factorization of $w_{i+1}$ to $\tilde x_1\tilde x_2$ such that $\tilde x_1$ leads from $u_i'$ to $\tilde u$ in $t_z$.
  Consequently, $\tilde x_1$ leads from $v_{i}'$ to $\hat v$ in $t^{\mathit{ann}}_z$ for some $\hat v$ in $t^{\mathit{ann}}_z$ such that $\mathit{ann}(\hat v) = \tilde u$.
  Let $\tilde x_1 = a_1\dots a_{m'}$ for some $m' < m$.
  Thus, after reading $a_1\dots a_{m'}$ output would be produced.
  Therefore, after reading $a_1\dots a_m$ the node $v_{i+1}$ would not be reached, because $\hat v$ and $v_{i+1}$ refer to the same number of produced output blocks, and thus $v_{i+1}$ is not reachable from $\hat v$ in $t^{\mathit{ann}}_z$. 
  Contradiction.
  Hence, there exists some $w' \in \Sigma_\inp^*$ such that $w'$ leads from $u_{i}'$ to $u_{i+1}$ it holds $w'_1o_1\dots w'_io_iw' \in \mathit{Pref}(U_{s})$ meaning condition $5$. is satisfied.
 \end{description}
 We have proved the claim of the induction.
 It is left to show that $xy\bar w_1o_1\dots \bar w_io_i\bar w_{i+1}\tilde z \in \llbracket S \rrbracket$ if, and only if, $xzw_1o_1\dots w_io_i\bar w_{i+1}\tilde z \in \llbracket S \rrbracket$ for all $\tilde z \in \Sigma_{\inp\outp}^*$ implying that $U'$ is a $T$-controlled uniformization.

 Recall, $\delta_\mathcal A^*(q_0^\mathcal A,(\pi_\inp(x),\pi_\outp(x))) = q$ and $\delta_\mathcal B^*(q_0^\mathcal B,xy) = p$, since $\tau_y = \tau_z$, also $\delta_\mathcal B^*(q_0^\mathcal B,xz) = p$.
 Now, we show that $\delta_\mathcal A^*(q,(w_1\dots w_
 {i+1},z))$ $=$ $\delta_\mathcal A^*(q,(\bar w_1\dots \bar w_{i+1},y))$ and $\delta_\mathcal B^*(p,w_1o_1\dots w_io_iw_{i+1})$ $=$ $\delta_\mathcal B^*(p,\bar w_1o_1\dots \bar w_io_i\bar w_{i+1})$.
 This is easy to see, because $w_1o_1\dots w_io_iw_{i+1}$ leads from $\ro{t^{\mathit{ann}}_z}$ to a leaf $v_{i+1}$ in the annotated STT $t^{\mathit{ann}}_z$ and $\bar w_1o_1\dots \bar w_io_i\bar w_{i+1}$ leads from $\ro{t_y}$ to the leaf $\mathit{ann}(v_{i+1})$ in the STT $t_y$ and by Lemma \ref{lemma:basic} this implies that the induced state transformations are equal, i.e., $U'$ is an $T$-controlled uniformization.
 
 Furthermore, given a word from $U'$, the first output block (after the lag has exceeded $\gamma$ at some point) increases the output lag by at most $b$.
 Applying this construction a total of $\lceil n/2 \rceil$ times yields a uniformization according the statement of the lemma. 
\end{proof}

The proof of the above lemma yields that $b$ can be chosen as $\mathit{max}\{\beta,\gamma+1\}$, where $\beta$ is the smallest bound on the length of representatives of output profiles.

The second auxiliary lemma states that it suffices to consider uniformizers that either produce output before the input sequence contains an idempotent factor, or if they do not produce output until then, then neither do they when pumping the idempotent factor.

\begin{lemma}\label{lemma:pumping}
 If $S$ has a $T$-controlled uniformization $U$ by an sDFA, then $S$ has a $T$-controlled uniformization $U'$ by an sDFA such that for each $u \in \Sigma_{\inp\outp}^*$ and $x,x_1,x_2 \in \Sigma_\inp^*$ such that $|\pi_\inp(u)| = |\pi_\outp(u)|$, $|xx_1x_2| > \gamma$, $x_2$ is idempotent, and $P_{x_1} = P_{x_2}$ it holds that if $uxx_1x_2 \in \mathit{Pref}(U')$, then $uxx_1x_2^i \in \mathit{Pref}(U')$ for each $i \in \mathbbm N$.
\end{lemma}

\begin{proof}
Let $U$ be a $T$-controlled uniformization recognized by an sDFA $\mathcal U$ such that there is $u \in \Sigma_{\inp\outp}^*$ and $x,x_1,x_2 \in \Sigma_\inp^*$ such that $|\pi_\inp(u)| = |\pi_\outp(u)|$, $|xx_1x_2| > \gamma$, $x_2$ is idempotent, and $P_{x_1} = P_{x_2}$ such that $uxx_1x_2 \in \mathit{Pref}(U)$ and $uxx_1x_2^i \notin \mathit{Pref}(U)$ for some $i \in \mathbbm N$.

Since $uxx_1x_2^i \notin \mathit{Pref}(U)$ there exists some $j < i$ and some prefix of $x_2$, say $x_2'$, such that $uxx_1x_2^jx_2' \in \mathit{Pref}(U)$ and $\delta_\mathcal U^*(q_0^\mathcal U,uxx_1x_2^jx_2') \in Q_\mathcal U^\outp$, i.e., $\mathcal U$ produces output after reading $uxx_1x_2^jx_2'$.
Now we show that there exists an $T$-controlled uniformization $U'$ recognized by an sDFA $\mathcal U'$ such that $\mathcal U'$ produces output after reading $uxx_1x_2'$.

Since $x_2$ is idempotent and $P_{x_1} = P_{x_2}$, then also $P_{x_1} = P_{x_1x_2^j}$.
Thus, similar as in the proof of Lemma \ref{lemma:shortregular}, we can show that $\mathcal U'$ can sequentially determine the output that has to be produced in a computation on $uxx_1x_2'w$ for some $w \in \Sigma_\inp^*$ by behaving like $\mathcal U$ on $uxx_1x_2^jx_2'w$ and replacing outputs that consume $x_1x_2^j$ (w.r.t.\ $\mathcal A$) by equal outputs that consume overlap $x_1$ (w.r.t.\ $\mathcal A$) in the sense that the induced state transformations on $\mathcal A$ and $\mathcal B$ are equal. 
\end{proof}

Recall, in Asm.~\ref{asm:bounds}, we have fixed bounds $r_1$ as in Lemma \ref{lemma:ramsey} and $r_2$ as in Lemma \ref{lemma:longeroutput}.
Also, $r_1,r_2 > \gamma$.
For a uniformizer according to Lemma~\ref{lemma:longeroutput}, the lemma yields that the next output block is of length at most $r_2$ if there is currently lag caused by output that is behind.
However, if there is currently lag because the output is behind, say $\ell$ symbols, then Lemma~\ref{lemma:longeroutput} yields that the next output block is of length at most $\ell+r_2$.
This value can become arbitrary large as the lag can generally not be bounded.
Our goal is to show that if lag caused by output that is behind exceeds $r_1$, then the length of the next output block can be bounded by $r_1$.
Recall, since $T \subseteq T_{\leq \gamma} \cdot (\Sigma_\inp^*+\Sigma_\outp^*)^n$ and $r_1 > \gamma$, it can only happen $\lceil n/2 \rceil$ times that lag caused by output that is behind exceeds $r_1$.

Thus, proving the above statement then gives us the key lemma, stated below.
Recall, $T_i$ is defined as $T \cap \left( T_{\leq \gamma} \cdot (\Sigma_\inp^*+\Sigma_\outp^{\leq i})^n \right)$ for $i \geq 0$.

{
\renewcommand{\thetheorem}{\ref{lemma:shortregular}}
\lemmashortregular*
\addtocounter{theorem}{-1}
}

\begin{proof}[Proof of Lemma~\ref{lemma:shortregular}]
 We now show that it is sufficient to consider $T_k$ for $k = r_1+r_2$ in order to find an $T$-controlled $\llbracket S \rrbracket$-uniformization if there exists one.
 
 Let $U$ be an $T$-controlled uniformization that satisfies the conditions of Lemmata  \ref{lemma:longeroutput} and \ref{lemma:pumping} recognized by a sequential DFA $\mathcal U$.
 We show how we can obtain $U'$ that is an $M_k$-controlled $\llbracket S \rrbracket$-uniformization recognized by a sequential DFA $\mathcal U'$ by modifying $\mathcal U$.
 
 Recall, $T \subseteq T_{\leq \gamma} \cdot (\Sigma_\inp^*+\Sigma_\outp^*)^n$.
 We construct $\mathcal U'$ such that a computation of $\mathcal U'$ differs from $\mathcal U$ from the point on that the following situations occurs:
 There is $u \in \Sigma_{\inp\outp}^*$ and $x,x_1,x_2 \in \Sigma_\inp^*$ such that $|\pi_\inp(u)| = |\pi_\outp(u)|$, $|xx_1x_2| > \gamma$, $x_2$ is idempotent, $P_{x_1} = P_{x_2}$ and $uxx_1x_2 \in \mathit{Pref}(U)$.
 
 Since $|\pi_\inp(u)| = |\pi_\outp(u)|$ and $|xx_1x_2| > \gamma$, we know that $uxx_1x_2 \in \mathit{Pref}(L)$, but $uxx_1x_2 \notin \mathit{Pref}(L_1)$.
 Thus, for each $z \in \Sigma_{\inp\outp}^*$ with $uxx_1x_2z \in U$ holds that $\mathit{shift}(z) < n$ and the number of output blocks in $z$ is at most $\lceil n/2 \rceil$.
 Let $\ell$ be the maximal size of an output block that $\mathcal U$ can produce.
 Hence, $|\pi_\outp(z)|$ is at most $\lceil n/2 \rceil \ell$.
 We chose the the smallest $m$ such that $|xx_1x_2^m| \geq \lceil n/2 \rceil \ell$.

 Consider an arbitrary input word $w \in \Sigma_\inp$ such that $\pi_\inp(u)xx_1x_2w$ is in the domain of $\llbracket S \rrbracket$. 
 In order to determine the computation of $\mathcal U'$ on $\pi_\inp(u)xx_1x_2w$, we consider the computation of $\mathcal U$ on $\pi_\inp(u)xx_1x_2^{m}w$.
 Note that by Lemma \ref{lemma:pumping}, after having read $\pi_\inp(u)xx_1x_2$, $\mathcal U$ produces no output while reading $x_2^{m-1}$.
 Assume we have already have defined $\mathcal U'$ up to the point where $\pi_\inp(u)xx_1x_2$ was read, and until then, $\mathcal U$ and $\mathcal U'$ have worked the same way.
 We show that $\mathcal U'$ can continue the computation successfully.
 
 Let $\delta_\mathcal U^*(q_0^\mathcal U,\pi_\inp(u)xx_1x_2^{m}) = s$, and let $\delta_\mathcal U'^*(q_0^{\mathcal U'},\pi_\inp(u)xx_1x_2) = r$.
 The computation of $\mathcal U'_r$ on $w$ will be based on the computation of $\mathcal U_s$ on $w$ such that we can sequentially define the output blocks that $\mathcal U'_r$ has to produce.
 
 Since $x_2$ is idempotent and $P_{x_1} = P_{x_2}$, also $P_{x_1x_2} = P_{x_1x_2^m}$.
 Furthermore, also $P_{xx_1x_2} = P_{xx_1x_2^m}$.
 Hence, we have $\mathrm{SST}^{\lceil n/2 \rceil}_{xx_1x_2} = \mathrm{SST}^{\lceil n/2 \rceil}_{xx_1x_2^m}$.
 We are interested in the unique tree $t_{xx_1x_2^m} \in \mathrm{SST}^{\lceil n/2 \rceil}_{xx_1x_2^m}$ with  root label $(p,q)$, where
 \begin{itemize}
  \item $p \in Q_\mathcal B$ is the state such that $\delta_\mathcal B^*(q_0^\mathcal B,uxx_1x_2^m) = p$, and
  \item $q \in Q_\mathcal A$ is the state such that $\delta_\mathcal A^*(q_0^\mathcal B,\pi_\inp(uxx_1x_2^mw_0),\pi_\outp(uxx_1x_2^mw_0)) = q$, where $w_0$ is the prefix of the input sequence $w$ such that either $\mathcal U_s$ produces output after reading $w_0$, or $w_0 \in U_{s}$, i.e., $w = w_0$ and no output was produced.
 \end{itemize}
 Note that, since $P_{xx_1x_2} = P_{xx_1x_2^m}$ we have $\tau_{xx_1x_2} = \tau_{xx_1x_2^m}$, thus also $\delta_\mathcal B^*(q_0^\mathcal B,uxx_1x_2) = p$.
 Let $t_{xx_1x_2}$ denote the same tree w.r.t.\ $xx_1x_2$, which has to exist since $P_{xx_1x_2} = P_{xx_1x_2^m}$. 
 Now, we are ready to inductively (on the number of produced output blocks) define the computation of $\mathcal U'_r$ (on $w$ as above) satisfying the following properties:
 
 For $w_0 \in \Sigma_\inp^*$, $w_1,\dots, w_{i-1} \in \Sigma_\inp^+$, $y_1,\dots,y_{i-1} \in \Sigma_\outp^+$ and $y_i \in \Sigma_\outp^*$ such that $w_0y_1 \dots w_{i-1}y_i \in \mathit{Pref}(U'_r)$, there exists a path $v_0'v_1\dots v_{i-1}'v_i$ in $t_{xx_1x_2}$ with $v_0' = \ro{t_{xx_1x_2}}$ and $v_i$ is a leaf in $t_{xx_1x_2}$ such that
 \begin{enumerate}
  \item $y_j$ leads from $v_{j-1}'$ to $v_j$ in $t_{xx_1x_2}$ for all $1 \leq j \leq i$, and
  \item $w_j$ leads from $v_j$ to $v_{j}'$ in $t_{xx_1x_2}$ for all $1 \leq j < i$, and
  \item there exist $\bar y_1,\dots \bar y_{i-1} \in \Sigma_\outp^+$ and $\bar y_i \in \Sigma_\outp^*$ such that
  \begin{enumerate}
   \item $\bar y_j$ leads from $v_{j-1}'$ to $v_j$ in $t_{xx_1x_2^m}$ for all $1 \leq j \leq i$, and
   \item $w_j$ leads from $v_j$ to $v_{j}'$ in $t_{xx_1x_2^m}$ for all $1 \leq j < i$, and
   \item $w_0\bar y_1\dots w_{i-1}\bar y_i \in \mathit{Pref}(U_{s})$.
  \end{enumerate}
 \end{enumerate}
 To be clear, the formulation $y_j$ leads from $v'_{j-1}$ to $v_j$ in $t_{xx_1x_2}$ is used to mean that $y_j$ leads from $v'_{j-1}$ to $v_j$ w.r.t.\ $x''$ and $j'$, where $j' = \lceil n/2 \rceil - (j-1)$ and $x'' \in \Sigma_\inp^*$ such that $xx_1x_2$ has a factorization $x'x''$ with $|x'| = |y_1\dots y_{j-1}|$.
 Analogously, the formulation $\bar y_j$ leads from $v'_{j-1}$ to $v_j$ in $t_{xx_1x_2^m}$ is used to mean that $\bar y_j$ leads from $v'_{j-1}$ to $v_j$ w.r.t.\ $x''$ and $j'$, where $j' = \lceil n/2 \rceil - (j-1)$ and $x'' \in \Sigma_\inp^*$ such that $xx_1x_2^m$ has a factorization $x'x''$ with $|x'| = |\bar y_1\dots \bar y_{j-1}|$.
 
 Assume $k = 1$, we already defined $w_0$ above and $v_0'$ as $\ro{t_{xx_1x_2}}$, we have to define $y_1$, $\bar y_1$, and $v_1$.
 The sequence $w_0$ was chosen such that either $\mathcal U_s$ produces output after reading $w_0$, or $w_0 \in U_s$ , i.e., the input sequence ends.
 In the former case, let $\bar y_1$ be the output that is produced by $\mathcal U_s$, in the latter case let $\bar y_1$ be $\varepsilon$.
 We chose as vertex $v_1$ the vertex $v$ such that $\bar y_1$ leads from $v_0'$ to $v$ in $t_{xx_1x_2^m}$.
 This vertex has to exist, because $1 \leq \lceil n/2 \rceil$ and and by choice of $m$ we have $|\bar y_1| \leq |xx_1x_2^m|$.
 Since $t_{xx_1x_2^m} = t_{xx_1x_2}$, there exists some $z \in \Sigma_\outp^*$ such that $z$ leads from $v_0'$ to $v$ in $t_{xx_1x_2}$, let $y_1$ be such a $z$ and let $w_0y_1 \in \mathit{Pref}(U'_r)$.
 Clearly, the above conditions are satisfied.
 
 Assume we have already defined the computation for some $k < i$ satisfying the above conditions, i.e., for $w_1,\dots,w_{k-1} \in \Sigma_\inp^+$ and $y_1,\dots,y_k \in \Sigma_\outp^+$ such that $w_0y_1 \dots w_{k-1}y_k \in \mathit{Pref}(U'_r)$, there exists a path $v_0'v_1\dots v_{k-1}'v_k$ in $t_{xx_1x_2}$ with $v_0' = \ro{t_{xx_1x_2}}$, and $y_1 \dots w_{k-1}y_k$ leads from $v_0'$ to $v_k$ in $t_{xx_1x_2}$, and there are $\bar y_1,\dots,\bar y_k \in \Sigma_\outp^+$ such that $\bar y_1\dots w_{k-1}\bar y_k$ leads from $v_0'$ to $v_k$ in $t_{xx_1x_2^m}$ and $w_0\bar y_1\dots w_{k-1}\bar y_k \in \mathit{Pref}(U_{s})$.
 
 Let $\delta_\mathcal U(s,w_0\bar y_1\dots w_{k-1}\bar y_k) = s_k$.
 To define $w_k$, $y_{k+1}$, $\bar{y}_{k+1}$, $v_k'$ and $v_{k+1}$, we consider the computation of $\mathcal U_{s_k}$, let $w_k \in \Sigma_\inp^+$ be the next part of the input sequence such that either $\mathcal U_{s_k}$ produces output after reading $w_k$, or $w_k \in U_{s_k}$, i.e., the input sequence ends.
 Let $v_{k}'$ be the vertex that is reached from $v_k$ via $w_k$ in $t_{xx_1x_2^m}$, note then $w_k$ also leads from $v_k$ to $v_k'$ in $t_{xx_1x_2}$.
 If after reading $w_k$ output is produced, then let $\bar y_{k+1}$ be the output that is produced by $\mathcal U_{s_k}$, if the input sequence ends, then let $\bar y_{k+1}$ be $\varepsilon$.
 For $v_{k+1}$ we pick the vertex $v$ such that $\bar y_{k+1}$ leads from $v_k'$ to $v$ in $t_{xx_1x_2^m}$.
 This vertex has to exist, because $k < \lceil n/2 \rceil$ and by choice of $m$ we have $|\bar y_1\dots \bar y_{k+1}| \leq |xx_1x_2^m|$.
 Since $t_{xx_1x_2^m} = t_{xx_1x_2}$, there exists some $z \in \Sigma_\outp^*$ such that $z$ leads from $v_k'$ to $v$ in $t_{xx_1x_2}$, let $y_{k+1}$ be such a $z$ and let $w_0y_1 \dots w_{k-1}y_kw_ky_{k+1} \in \mathit{Pref}(U'_r)$.
 With these choices the conditions stated above are satisfied.
 
 Note that after at most $\lceil n/2 \rceil$ output blocks, $v_{k+1}$ will be a leaf, because eventually either $|\bar y_1\dots \bar y_{k+1}| = |xx_1x_2^m|$ or $|\bar y_1\dots \bar y_{k+1}| \leq |xx_1x_2^m|$ and the input sequence has ended.
 Both cases imply that $\bar y_1\dots \bar y_{k+1}$ fully traverses $xx_1x_2^m$ in $\mathcal A_q$, thus, by construction of state transformation trees, a leaf is reached.
 Note that this also implies that $|y_1\dots y_{k+1}| \leq |xx_1x_2|$.
 After reaching a leaf, $\mathcal U'$ continues to read the remaining input, say $w_i$, as $\mathcal U$ does.
 
 Altogether, from the induction above, it follows that for $uxx_1x_2w_0y_1\dots w_{i-1}y_iw_i \in U'$, there is $uxx_1x_2uxx_1x_2^mw_0\bar y_1\dots w_{i-1}\bar y_iw_i \in U$ such that both $y_1\dots w_{i-1}y_i$ and $\bar y_1\dots w_{i-1}\bar y_i$ lead through the same path in $t_{xx_1x_2} = t_{xx_1x_2^m}$, thus, we obtain $\delta_\mathcal B^*(p,y_1\dots w_{i-1}y_i)$ $=$ $\delta_\mathcal B^*(p,\bar y_1\dots w_{i-1}\bar y_i)$ and $\delta_\mathcal A^*(q,(xx_1x_2,y_1\dots y_i))$ $=$ $\delta_\mathcal A^*(q,(xx_1x_2^m,\bar y_1\dots \bar y_i))$.
 Together with $\tau_{xx_1x_2} = \tau_{xx_1x_2^m}$, it now directly follows that $uxx_1x_2w_0y_1\dots w_{i-1}y_iw_i$ is $L$-controlled and $\llbracket uxx_1x_2w_0y_1\dots w_{i-1}y_iw_i \rrbracket \in \llbracket S \rrbracket$.
 
 That means $U'$ is an $T$-controlled uniformization, we argue that it is even a $T_k$-controlled uniformization, recall $T_k = T \cap \left( T_{\leq \gamma} \cdot (\Sigma_\inp^*+\Sigma_\outp^{\leq k})^n \right)$, where $k = r_1+r_2$.
 As seen above, if the input lag is large enough that it contains an idempotent factor, which is given after at most $r_1$ input symbols, then the remaining output is of length at most $r_1 < k$.
 If the lag is smaller then $r_1$ input symbols, say $d$, and output is produced, then Lemma \ref{lemma:longeroutput} yields that the produced output block is of length at most $d+r_2 < k$.
 Hence $U'$ is indeed $T_k$-controlled. 
\end{proof}

The proof of the above lemma yields that we can focus on the construction of $T_k$-controlled uniformizer, where $k$ can be chosen as $r_1+r_2$.

\subsection*{Part II}

The goal of this section is to show that the problem whether $S$ has a $T_i$-controlled uniformizer reduces to the question whether $T_i(S)$ has a subset uniformizer for some suitable $T_i(S)$ as defined in Lemma~\ref{lemma:transformregular}.
Together with Lemma~\ref{lemma:shortregular} this then directly yields Theorem~\ref{thm:regular}.

{
\renewcommand{\thetheorem}{\ref{lemma:transformregular}}
\lemmatransformregular*
\addtocounter{theorem}{-1}
}

\begin{proof}[Proof of Lemma~\ref{lemma:transformregular}]

It is possible to give a direct construction from $S$ to $T_i(S)$, however, it is simpler to give a construction from $S_{\mathit{can}}$ to $T_i(S)$.

So, we work with $S_{\mathit{can}}$.
Let $\mathcal M$ be an NFA recognizing the regular set $T_i$.
From $\mathcal M$, we construct an NFA $\mathcal C$ that while reading an $T_i$-controlled word $w$ simultaneously constructs an $(12)^*(1^*+2^*)$-controlled word $w'$ with $\llbracket w \rrbracket = \llbracket w' \rrbracket$ and simulates $\mathcal A$ on $w'$ and accepts if $\mathcal A$ accepts $w'$.
In other words, $\mathcal C$ resynchronizes $w$ on the fly to be $(12)^*(1^*+2^*)$-controlled in order to simulate $\mathcal A$ on the resynchronization.
We only give an idea of the construction.

A $T_i$-controlled word $w$ can be factorized as $w_1 \cdot w_2$ such that $w_1 \in L_{\leq \gamma}$ and $w_2 \in (\Sigma_\inp^*+\Sigma_\outp^{\leq i})^n$.
For each word $w$, $\mathcal C$ guesses when the split occurs and uses different resynchronization techniques on $w_1$ and $w_2$.
We now describe how $w_1$ and $w_2$ are resynchronized.

While reading $w_1$, every position is at most $\gamma$-lagged.
To resynchronize $w_1$ to have a $(12)^*$-controlled synchronization (i.e., a $1$-lagged synchronization), $\mathcal C$ has to store only a window of $\gamma$ symbols to be able to continue the simulation of $\mathcal A$.

Concerning $w_2$, we use the following method to obtain a $w_2$ resynchronization that is 
$(12)^*(1^*+2^*)$-controlled.
The length of $\pi_{\outp}(w_2)$ is bounded by $i \cdot n$.
Hence, before reading $w_2$, $\mathcal C$ guesses an output word $y \in \Sigma_{\outp}^*$ of length at most $i \cdot n$.
While reading $w_2$, $\mathcal C$ can easily simulate $\mathcal A$ on the $(12)^*(1^*+2^*)$-controlled synchronization of $(\pi_{\inp}(w_2),y)$ and check whether $y = \pi_{\outp}(w_2)$.

It is easy to see that $\mathcal C$ indeed recognizes the desired language $T_i(S)$. 
\end{proof}

Now we have all ingredients to prove our main result.

{
\renewcommand{\thetheorem}{\ref{thm:regular}}
\thmregular*
\addtocounter{theorem}{-1}
}

\begin{proof}[Proof of Theorem~\ref{thm:regular} continued]
It is left to show $S$ has a $T$-controlled uniformization by an sDFA iff $\mathrm{dom}(\llbracket S \rrbracket) = \mathrm{dom}(\llbracket T_k(S) \rrbracket)$ and $T_k(S)$ has a subset uniformization by an sDFA, which is decidable by Theorem~\ref{thm:sunif}.

Assume $\mathrm{dom}(\llbracket S \rrbracket) = \mathrm{dom}(\llbracket T_k(S) \rrbracket)$ and $T_k(S)$ has a subset uniformization by an sDFA.
Since $\mathrm{dom}(\llbracket S \rrbracket) = \mathrm{dom}(\llbracket T_k(S) \rrbracket)$, every subset uniformization of $T_k(S)$ is also a $T_k$-controlled uniformization of $S$.
Such a uniformization is also $T$-controlled, because $T_k \subseteq T$.

For the other direction, assume $S$ has a $T$-controlled uniformization by an sDFA.
As stated above, then $S$ has a $T_k$-controlled uniformization by an sDFA, say $U \subseteq T_k$.
First, we show $\mathrm{dom}(\llbracket S \rrbracket) = \mathrm{dom}(\llbracket T_k(S) \rrbracket)$.
Since $\llbracket U \rrbracket \subseteq_\mathit{u} \llbracket S \rrbracket$, we have $\mathrm{dom}(\llbracket S \rrbracket) = \mathrm{dom}(\llbracket U \rrbracket)$.
Clearly, by construction, $\mathrm{dom}(\llbracket S \rrbracket) \supseteq \mathrm{dom}(\llbracket T_k(S) \rrbracket)$.
Since $U$ is $T_k$-controlled, also $U \subseteq T_k(S)$ and $\mathrm{dom}(\llbracket U \rrbracket) \subseteq \mathrm{dom}(\llbracket T_k(S) \rrbracket)$ by construction.
We can conclude $\mathrm{dom}(\llbracket S \rrbracket) = \mathrm{dom}(\llbracket T_k(S) \rrbracket)$.
Secondly, we show that $U$ is a subset uniformization of $T_k(S)$.
Since $U \subseteq T_k(S)$ and $\mathrm{dom}(\llbracket U \rrbracket) = \mathrm{dom}(\llbracket T_k(S) \rrbracket)$, it is clear that $U$ is a subset uniformization of $T_k(S)$.

This concludes the proof of the claim. 
\end{proof}

\end{document}